\documentclass[12pt,draftcls, onecolumn, dvips]{IEEEtran}
\usepackage[pdftex]{graphics}
\usepackage{amsmath}
\usepackage{times}
\usepackage{amsmath,dsfont}
\usepackage{amssymb,amsthm,bm}
\usepackage{epsfig,verbatim}
\usepackage{subfig}
\usepackage[utf8]{inputenc}
\usepackage[T1]{fontenc}
\usepackage{textcomp}
\usepackage{gensymb}
\usepackage{fancyhdr}

\newtheorem{definition}{Definition}
\newtheorem{theorem}{Theorem}

\newtheorem{corollary}{Corollary}
\newtheorem{proposition}{Proposition}

\newtheorem{remark}{Comment}

\DeclareMathOperator*{\argmax}{arg\,max}
\newcommand{\Tgood}{\mbox{Tx}_1}
\newcommand{\Tbad}{\mbox{Tx}_2}

\newcommand{\E}{\mathds{E}}
\newcommand{\mE}{\mathcal{E}}
\newcommand{\cov}{\mbox{cov}}
\newcommand{\C}{\mathds{C}}
\newcommand{\Amat}{\mathds{A}}
\newcommand{\Bmat}{\mathds{B}}
\newcommand{\Hmat}{\mathds{H}}

\newcommand{\Imat}{\mathds{I}}

\newcommand{\Yvec}{\mathbf{Y}}
\newcommand{\yvec}{\mathbf{y}}

\newcommand{\Xvec}{\mathbf{X}}
\newcommand{\xvec}{\mathbf{x}}
\newcommand{\Zvec}{\mathbf{Z}}

\newcommand{\Cset}{\mathfrak{C}}

\newcommand{\Rset}{\mathfrak{R}}

\renewcommand{\th}{\tilde{h}}
\newcommand{\utH}{\underline{\tilde{H}}}
\newcommand{\uth}{\underline{\tilde{h}}}
\newcommand{\utheta}{\underline{\theta}}

\newcommand{\CN}{\mathcal{CN}}

\newcommand{\mS}{\mathcal{S}}
\newcommand{\M}{\mathcal{M}}

\newcommand{\tH}{\tilde{H}}
\newcommand{\uH}{\underline{H}}
\newcommand{\uh}{\underline{h}}

\newcommand{\mT}{\mathcal{T}}

\newcommand{\eps}{\epsilon}
\newcommand{\SNRR}{\makebox{SNR}}
\newcommand{\SNR}{\rho}
\newcommand{\sSNR}{\scriptsize \rho}
\newcommand{\outage}{\mathcal{O}}
\newcommand{\code}{\mathcal{S}_c}
\renewcommand{\P}{\makebox{P}}
\newcommand{\Real}{\mathfrak{Re}}

\long\def\symbolfootnote[#1]#2{\begingroup\def\thefootnote{\fnsymbol{footnote}}\footnote[#1]{#2}\endgroup}
\newcommand{\tend}{\hfill$\blacksquare$}

\IEEEoverridecommandlockouts

\title{Diversity-Multiplexing Tradeoff for the Interference Channel with a Relay
\thanks{The material in this paper was presented in part at the IEEE International Symposium on Information Theory (ISIT), Jul. 2013, Istanbul, Turkey. This work was supported in part by the European Commission’s Marie Curie IRG Fellowship PIRG05-GA-2009-246657 under the Seventh Framework Programme, the DARPA ITMANET program under grant 1105741-1-TFIND, ONR grant N000140910072, the DTRA program under grant HDTRA1-08-1-0010, DoD with grant HDTRA1-13-1-0029, and the NSF Center for the Science of Information with grants CNS-1343155, ECCS-1305979, CNS-1265227, NSFC-61328102, and under Award CCF-0939370.}
}

\author{
\IEEEauthorblockN{ Daniel Zahavi$^\star$, {\em Student Member, IEEE}, Lili Zhang$^{\S}$, {\em Member, IEEE},\\ Ivana Maric$^{\dagger}$, {\em Member, IEEE}, Ron Dabora$^\star$, {\em Senior Member, IEEE},\\ Andrea J. Goldsmith$^\ddagger$, {\em Fellow, IEEE}, and Shuguang Cui$^{\dagger\dagger}$, {\em Fellow, IEEE}}\\

\vspace{0.3cm}

\authorblockA{\small $\star$ \!Ben-Gurion University \hspace{0.0cm} $\S$ \!Qualcomm Inc. \hspace{0.0cm} $\dagger$ \!Ericsson Research \hspace{0.0cm} $\ddagger$ \!Stanford University \hspace{0.0cm} $\dagger\dagger$ \!\!Texas A$\&$M University}\vspace{-0.5 cm}}

\begin{document}
\maketitle
\begin{picture}(0,0)
\put(-10,280){\tt\small Accepted to the IEEE Transactions on
Information Theory,  Dec 2014.}
\end{picture}

\begin{abstract}
    We study the diversity-multiplexing tradeoff (DMT) for the slow fading interference channel with a relay (ICR). We derive four inner bounds on the DMT region: the first is based on the compress-and-forward (CF) relaying scheme, the second is  based on the decode-and-forward (DF) relaying scheme, and the last two bounds are based on the half-duplex (HD) and full-duplex (FD) amplify-and-forward (AF) schemes. For the CF and DF schemes, we find conditions on the channel parameters and the multiplexing gains, under which the corresponding inner bound achieves the optimal DMT region. We also identify cases in which the DMT region of the ICR corresponds to that of two parallel slow fading relay channels, implying that {\em interference does not decrease} the DMT for each pair, and that a {\em single} relay can be DMT-optimal for two pairs {\em simultaneously}. For the HD-AF scheme we derive conditions on the channel coefficients under which the proposed scheme achieves the optimal DMT for {\em the AF-based} relay channel. Lastly, we identify conditions under which adding a relay strictly enlarges the DMT region relative to the interference channel without a relay.\vspace{-0.1cm}
\end{abstract}
\section{Introduction}
\label{sec_intro}
\vspace{-0.2cm}
The interference channel with a relay (ICR) is a canonical network model in which a relay helps two independent transmitters, Tx$_1$ and Tx$_2$, in sending messages to their corresponding receivers, Rx$_1$ and Rx$_2$, simultaneously over a shared channel. The ICR provides design insights and performance bounds on cooperation strategies for wireless networks with interference.

The ICR was first studied in \cite{Sahin:07} and has since been the focus of considerable research. Inner and outer bounds on the capacity region of the two-user ICR with additive white Gaussian noise (AWGN) were characterized in \cite{Sahin:07}, \cite{Tian:11}, and \cite{Maric:09}. In \cite{Sahin:07} an achievable region was obtained by employing a rate-splitting scheme at the transmitters, a decode-and-forward (DF) strategy at the relay, and a backward decoding scheme at the receivers. The work in \cite{Tian:11} used the compress-and-forward (CF) strategy at the relay to obtain an achievable rate region.
Outer bounds for the AWGN-ICR were obtained by applying the cut-set bound in \cite{Tian:11}, \cite{Maric:09}, and by using a potent relay in combination with genie-aided methods in \cite{Tian:11}. Capacity regions for ergodic phase-fading and for ergodic Rayleigh fading ICRs in the strong interference regime were characterized in \cite{Dabora:12} for the case in which the source-relay links are good, i.e., when DF achieves capacity.

The diversity-multiplexing tradeoff (DMT), first introduced in \cite{Zheng:03}, characterizes the fundamental tradeoff between rate and reliability for multiple-antenna channels in the high signal-to-noise ratio (SNR) regime. In general, employing multiple-input multiple-output (MIMO) schemes allows higher rates compared to single-antenna schemes. However, these schemes require the physical dimensions of the nodes in a network to be large enough such that (s.t.) it is practical to mount a multiple-antenna array on each node.

In recent studies it was shown that some of the benefits of MIMO from the DMT perspective can be gained through user cooperation rather than using physical multiple-antenna arrays. In \cite{Yuksel:07}, the DMT characteristics of several relaying configurations were derived for both half-duplex (HD) and full-duplex (FD) relaying in two scenarios: (1) A clustered scenario, in which the relay nodes are clustered either with the source or with the destination (the channel between the nodes in the cluster is modeled as an AWGN channel), and (2) A non-clustered scenario, in which all channel coefficients matrices have independent and identically distributed (i.i.d.) Rayleigh-distributed entries. Specifically, \cite{Yuksel:07} first studied the DMT of FD multiple-antenna single-relay channels and showed that while the DF scheme is optimal for single-antenna relay channels, it is suboptimal for multiple-antenna relay channels. On the other hand, CF was shown to be optimal for the MIMO relay channel over the whole range of multiplexing gains, and for both clustered and  non-clustered relay networks using either HD or FD relay nodes. Finally, \cite{Yuksel:07} compared different cooperation strategies and models against conventional MIMO schemes and concluded that in many scenarios relaying cannot achieve the same DMT as non-virtual MIMO systems. In \cite{Karmakar:12} it was shown that quantize-map-and-forward (QMF) achieves the optimal DMT for certain configurations of the HD relay channel without channel state information (CSI) at the relay node. DMT analysis of the two-hop two-way MIMO relay channel was presented in \cite{Gunduz:08}, which showed that in such a scenario CF at the relay is optimal. The work in \cite{Gunduz:08} also proposed a dynamic CF protocol for the one-way multi-hop MIMO relay channel with a HD relay node, and showed that this scheme achieves the optimal DMT.

The DMT of the single-antenna block Rayleigh fading interference channel (IC) was studied in \cite{Leveque:09} for the scenario in which CSI is available both at the receivers (Rx-CSI) and at the transmitters (Tx-CSI), and in \cite{Bolcskei:09} for the scenario with only Rx-CSI. In \cite{Bolcskei:09}, it was shown that in the very strong interference regime, successive decoding with interference cancellation is DMT optimal.
Since for the ergodic fading case, using the approach of \cite{Bolcskei:09} achieves the capacity region of the IC \cite{Dabora:12}, it follows that the same strategy is optimal from both DMT and capacity perspectives. Additionally, for general interference regimes, \cite{Bolcskei:09} proposed a transmission scheme using a Han-Kobayashi \cite{Han:81} type superposition encoding, in which each receiver jointly decodes the common messages from both transmitters, and the private message from its intended transmitter. This scheme was shown to be DMT optimal in the strong and in the very strong interference regimes, and over a certain range of multiplexing gains. The DMT region of the block Rayleigh fading Gaussian MIMO ICR was studied in \cite{Maric:11} for the case where all links have the same exponential behavior as a function of the SNR. In \cite{Maric:11}, an outer bound on the DMT  was derived using the cut-set theorem, and an achievable DMT was characterized using CF at the relay subject to probabilistic conditions. The generalized degrees of freedom (GDoF) of the ICR was studied in \cite{Chaaban:12} where new outer bounds on the achievable GDoF of the ICR were derived using genie-aided methods. The work \cite{Chaaban:12} also proposed a {\em functional decode-and-forward} (FDF) strategy which was shown to achieve the optimal GDoF of the ICR given some constraints on the gains of the links in the channel. It was also shown in \cite{Chaaban:12} that relaying can increase the GDoF compared to the IC.

\vspace{-0.5cm}

\subsection*{Main Contributions}
\vspace{-0.2cm}
In this work we present results on the DMT region of the single-antenna ICR considering both FD and HD relaying. We consider the scenario in which the receivers have perfect Rx-CSI, but there is no Tx-CSI at the sources. We allow the direct link gains, the interfering link gains, the source-relay link gains, and the relay-destinations link gains to scale differently as exponential functions of the $\SNRR$, and focus  on characterizing the effects of the relationship between interference and cooperation on the DMT region. The channel model is symmetric in the sense that the scaling of the corresponding links at both pairs is identical.

The main contributions of this work are:
\begin{enumerate}
    \item Four achievable DMT regions are derived based on the CF, DF, and the HD and FD AF relaying schemes. For each scheme, we analyze the effect of the cross-link gains (interference), the relay-destination link gains (cooperation), and the source-relay links gains on the achievable DMT region.
    \item We derive sufficient conditions under which each of the schemes DF and CF achieves the {\em optimal DMT}. Based on the optimality results, we further obtain sufficient conditions under which the ICR has the same DMT as that of two parallel single-relay channels. Thus, the relay assistance to one pair {\em does not degrade} the DMT performance at the other pair, and {\em a single relay is simultaneously DMT-optimal for two separate communicating pairs sharing one channel}.
    \item For the AF scheme, we derive conditions under which each pair in the ICR has a DMT which is equal to the {\em best known} DMT for the {\em relay channel} with AF relaying.
    \item We compare the DMT of the ICR with that of the IC, and provide sufficient conditions under which {\em adding a relay to the IC strictly increases the DMT region}.
\end{enumerate}
The rest of the paper is organized as follows: The channel model and notation are presented in Section \ref{sec:Model}. The DMT performance with CF at the relay is studied in Section \ref{sec:DMTCF}, and the DMT performance with DF relaying is studied in Section \ref{sec:DMTDF}. AF relaying is studied in Section \ref{sec:DMTAF} for the HD and the FD regimes. Concluding remarks are presented in Section \ref{sec:summary}.

\section{Notation and System Model}
\label{sec:Model}

\vspace{-0.2 cm}

In the following we denote random variables (RVs) with upper-case letters, e.g., $X$, $Y$, and their realizations with lower-case letters, e.g., $x,y$. We denote the probability density function (p.d.f.) of a continuous complex-valued RV $X$ with $f_{X}(x)$. For brevity, the subscript $X$ may be omitted when it is the upper-case version of the realization symbol $x$. Upper-case double-stroke letters are used for denoting matrices, e.g., $\Amat$, with the exception that $\E\{X\}$ denotes the stochastic expectation of $X$. $\Imat_m$ denotes the $m\times m$ identity matrix. Bold-face letters, e.g., $\mathbf{x}$, denote column vectors (unless otherwise specified), and the $i$'th element of a vector $\mathbf{x}$ is denoted by $x_i$. We use $x^j$  to denote the vector $(x_1, x_{2},...,x_{j-1},x_j)$, $X^*$ to denote the conjugate of $X$, $\Amat^H$ to denote the Hermitian transpose of $\Amat$, and  $(x)^+$  to denote $\max \left\{x,0\right\}$. Given two $n\times n$ Hermitian matrices, $\Amat, \Bmat$, we write $\Bmat\preceq\Amat$ if $\Amat-\Bmat$ is positive semidefinite (p.s.d.) and $\Bmat\prec\Amat$ if $\Amat-\Bmat$ is positive definite (p.d.). $A^{(n)}_\epsilon(X)$ denotes the set of weakly jointly typical sequences with respect to $f_{X}(x)$, as defined in \cite[Sec. 8.2]{cover-thomas:it-book}. We denote the circularly symmetric, complex Normal distribution with mean $\mu$ and variance $\sigma^2$ with $\CN(\mu,\sigma^2)$, and the set of complex numbers with $\Cset$. Lastly, we denote $f(\SNR)\doteq \SNR^c$ if $\lim_{\substack{\sSNR\to\infty}}\frac{\log f(\SNR)}{\log\SNR}=c$. Given $f(\SNR)\doteq \SNR^c$ and $g(\SNR)\doteq \SNR^d$, we write $f(\SNR)\dot{\le}g(\SNR)$ if $c\le d$. All logarithms are to base $2$.

The ICR consists of two transmitters, two receivers, and a relay node that assists communications from the transmitters to their
respective receivers, as shown in Fig. \ref{fig:ICR}. Tx$_k$ sends messages~to~Rx$_k$,~$k=1,2$. The received signals at~Rx$_1$, Rx$_2$ and at the relay at time $i$ are denoted by $Y_{1,i}$, $Y_{2,i}$, and $Y_{3,i}$ respectively. The channel inputs from $\Tgood$, $\Tbad$ and the relay at time $i$ are denoted by $X_{1,i}$, $X_{2,i}$, and $X_{3,i}$, respectively. Let $\SNR$ denote the average received SNR over the direct link for both pairs, and let $H_{kl}$ denote the normalized channel coefficient from node $k$ to node $l$, s.t. its variance is normalized to one. We assume that all normalized channel coefficients are generated i.i.d. according to $\CN(0,1)$. The relationship between the channel inputs and outputs at time $i$, $i= 1$, $2,...$, $n$, is characterized by the following equations:
\vspace{-0.1cm}
\begin{subequations}
\label{eq:ICR_model}
\begin{eqnarray}
    \label{eq:channel_model_1}
    Y_{1,i} &=& \sqrt{\SNR}H_{11}X_{1,i}+\sqrt{\SNR^{\alpha}}H_{21}X_{2,i}+\sqrt{\SNR^{\beta}}H_{31}X_{3,i}+Z_{1,i},\\
    \label{eq:channel_model_2}
    Y_{2,i} &=& \sqrt{\SNR^{\alpha}}H_{12}X_{1,i}+\sqrt{\SNR}H_{22}X_{2,i}+\sqrt{\SNR^{\beta}}H_{32}X_{3,i}+Z_{2,i},\\
    \label{eq:channel_model_3}
    Y_{3,i} &=& \sqrt{\SNR^{\gamma}}H_{13}X_{1,i}+\sqrt{\SNR^{\gamma}}H_{23}X_{2,i}+Z_{3,i},
\end{eqnarray}
\end{subequations}
\begin{figure} [t]
    \centering
    \includegraphics[scale=0.27]{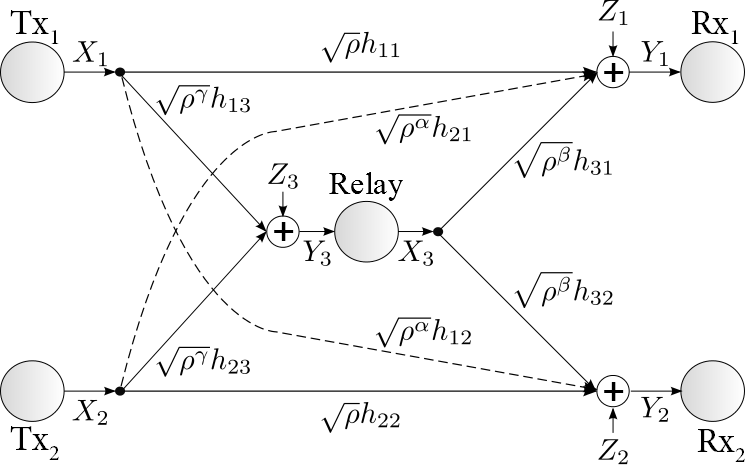}
    \vspace{-0.3cm}
    \caption{\small The ICR model. $\SNR$ denotes the SNR.}
\label{fig:ICR}
\vspace{-1cm}
\end{figure}
Here, $Z_{1,i}$, $Z_{2,i}$, and $Z_{3,i}$ are mutually independent RVs, distributed according to $\CN(0,1)$, independent over time and independent of the channel inputs and of the normalized channel coefficients. Each channel input has a per-symbol unit power constraint: $\E\{|X_{k,i}|^2\}\leq 1, k\in\{1,2,3\}$. The normalized channel coefficients are generated at the beginning of the codeword, and once determined, they remain fixed throughout the transmission of the entire codeword, which corresponds to slow fading.

Note that in \eqref{eq:ICR_model}, the cross-link gains scale as $\sqrt{\SNR^\alpha}$, the relay-destination link gains scale as $\sqrt{\SNR^\beta}$, the source-relay link gains scale as $\sqrt{\SNR^\gamma}$, and the source-destination link gains scale as $\sqrt{\SNR}$. This allows us to identify the impact of the scaling of the links in the channel on the DMT. The SNR exponents in \eqref{eq:ICR_model} define the ratio of the strength of different links in decibels. This model is very common in high SNR analysis of interference networks, see also \cite{Bolcskei:09}, \cite{Tse:08}, \cite{Jafar:10}, and \cite{Chaaban:12}. From a physical point of view, the different exponents represent different pathloss scaling behaviour due to different propagation conditions. The values of these exponents depend on the geographical setup and on whether a line-of-sight (LOS) exists between a transmitter and a receiver. Specifically, let $l$ denote the distance between a transmitter and a receiver. Then, as observed, e.g. in \cite{Molisch:06}, the received signal power at the receiver (and therefore the SNR\footnote{In this work we fix the power of the additive white Gaussian noise to one.}) scales proportionally to $l^{-\eta}$, where $\eta$ is commonly referred to as the {\em pathloss exponent}. Different studies show that the LOS pathloss exponents in indoor environments range from $1$ to $2$, while non-LOS pathloss exponents typically range from $3$ to $7$ \cite{Molisch:06}.

\noindent
The following CSI assumptions are made for the different schemes studied in this work:
\begin{itemize}
    \item In the study of CF relaying in Section \ref{sec:DMTCF}, it is assumed that Rx$_k$ has Rx-CSI on all of its incoming links. This Rx-CSI is represented by $\tH_k \triangleq \big( H_{1k}, H_{2k}, H_{3k}\big)\in\Cset^3\triangleq\tilde{\mathfrak{H}}_k, k\in\{1,2\}$. It is also assumed that the relay has Rx-CSI  represented by $\tH_{3}\triangleq(H_{13},H_{23})\in\Cset^2\triangleq\mathfrak{H}_3$. {\em In this section only}, we assume that the relay also has Tx-CSI on its outgoing links represented by $\tH_{3,T}\triangleq(H_{31},H_{32})\in\Cset^2$.
    \item In the study of DF relaying in Section \ref{sec:DMTDF}, it is assumed that the receivers and the relay have only Rx-CSI, each on its incoming links. The Rx-CSI of Rx$_k$ is $\tH_k, k=1,2$, and the Rx-CSI of the relay is $\tH_{3}$.
    \item In the study of AF relaying in Section \ref{sec:DMTAF}, it is assumed that the relay has only Rx-CSI, $\tH_{3}$, and that Rx$_k, k=1,2$ has Rx-CSI on its incoming links and on the incoming links of the relay, i.e., $(\tH_k,\tH_{3})$.
    \item Throughout this work, there is no Tx-CSI at Tx$_1$ or at Tx$_2$.
\end{itemize}
Finally, we let $\utH\triangleq(\tH_1,\tH_2,\tH_3,\tH_{3,T})$ be the vector of all channel coefficients.

\vspace{-0.3cm}

\begin{definition}
\label{def:code}
    An $(R_1,R_2,n)$ code \em{}for the slow-fading ICR consists of two message sets $\M_k \triangleq \big\{1,2,...,2^{n R_k}\big\}$, $k = 1,2$, two encoders at the sources,   $e_k: \M_k \mapsto \Cset^n$, $k=1,2$, and two decoders at the destinations,  $g_k: \tilde{\mathfrak{H}}_k \times \Cset^n \mapsto \M_k$, $k=1,2$. With Rx-CSI only, the transmitted signal at the relay  at time $i$ is $x_{3,i} = t_i\big(y_{3,1}^{i-1},\th_{3}\big) \in \Cset$, $i=1,2,...,n$. With both Rx-CSI and Tx-CSI at the relay we have $x_{3,i} = t_i\big(y_{3,1}^{i-1},\th_{3},\th_{3,T}\big) \in \Cset$. We denote a coding scheme by $\code$.
\end{definition}

\begin{definition}
The average probability of error \em{}where each sender selects its message independently and uniformly from its message set is $\P_e^{(n)} \triangleq \Pr\big(g_1(\tH_1, Y_1^n) \ne M_1 \mbox{ or } g_2(\tH_2, Y_2^n) \ne M_2\big)$.
\end{definition}

\begin{definition}
\em{} A rate pair $(R_1, R_2)$ is called {\em achievable} if for any $\epsilon\! >\!0$ and $\delta\! >\!0$ there exists some blocklength $n_0(\epsilon,\delta)$ s.t. for every $n\! >\! n_0(\epsilon,\delta)$ there exists an  \!$(R_1\! -\! \delta, R_2\! -\! \delta ,n)$\! code with $\P_e^{(n)}\! <\! \epsilon$. $\mathcal{R}(\uth,\code,\SNR)$ denotes the maximum achievable rate region, achieved by a coding scheme $\code$ for the ICR whose channel coefficients are $\uth$, and the direct-link SNR is $\SNR$.
\end{definition}

\begin{definition}
\em{} {\em The probability of an outage} event in the slow-fading ICR, for the scheme $\code$ and target rates $R_{1,T}$ for pair Tx$_1$-Rx$_1$, and $R_{2,T}$ for pair Tx$_2$-Rx$_2$, is defined as:
    \begin{equation*}
      \label{eq:Outage probability of ICR}
      P_\outage(R_{1,T}, R_{2,T},\SNR,\code)\triangleq \Pr\big((R_{1,T},R_{2,T})\notin\mathcal{R}(\utH,\code,\SNR)\big).
    \end{equation*}
\end{definition}

\begin{definition}
    \em{}We say that a coding scheme $\code$ for the ICR achieves {\em multiplexing gains} $(r_1, r_2)$, if there exist rates $\big(R_1(\uth,\code,\SNR),R_2(\uth,\code,\SNR)\big)\in\mathcal{R}(\uth,\code,\SNR)$ that scale as (see, e.g., \cite[Section III]{Yuksel:07})
    \begin{equation*}
        \lim_{\footnotesize{\SNR\rightarrow\infty}}\frac{R_k(\uth,\code,\SNR)}{\log(\SNR)}=r_k, \qquad k=1,2.
    \end{equation*}
\end{definition}

\begin{definition}
\label{Def:DG}
    \em{}We say that a scheme $\code$ achieves a diversity gain  of $d(r_1,r_2)$ for multiplexing gains $(r_1, r_2)$, if (see, e.g., \cite[Section III]{Yuksel:07})
    \begin{equation*}
        \label{eq:DG definition}
        - \lim_{\footnotesize{\SNR\rightarrow\infty}}\frac{\log P_\outage\Big(r_1\log(\SNR), r_2\log(\SNR),\SNR,\code\Big)}{\log(\SNR)}=d(r_1,r_2).
    \end{equation*}
    \noindent
\end{definition}

\section{Rx-CSI and Tx-CSI at the Relay: Compress-and-Forward}
\label{sec:DMTCF}
\subsection{An Outer Bound on the DMT Region}
\label{sec:OB}
When the receivers have Rx-CSI and the relay has
Rx-CSI  and Tx-CSI we have the following outer bound on the DMT region:
\vspace{-0.5cm}
\begin{proposition}
\label{thm:theorem_cutset}
    An outer bound on the DMT region of the symmetric ICR  is given by
    \begin{equation}
        \label{eq:cut_upper}
        d^{+}(r_1,r_2) = \min_{k\in\{1,2,3,4\}}\big\{d_k^+(r_1,r_2)\big\},
    \end{equation}

    \noindent
    where,
    \begin{subequations}
    \label{eq:cut_dmt}
    \begin{eqnarray}
        \label{eq:cut_dmt1}d^+_1(r_1,r_2)  &=& (1-r_1)^++(\gamma-r_1)^+\\
        \label{eq:cut_dmt2}d^+_2(r_1,r_2)  &=& (1-r_1)^++(\beta-r_1)^+\\
        \label{eq:cut_dmt3}d^+_3(r_1,r_2)  &=& (1-r_2)^++(\gamma-r_2)^+\\
        \label{eq:cut_dmt4}d^+_4(r_1,r_2)  &=& (1-r_2)^++(\beta-r_2)^+
    \end{eqnarray}
    \end{subequations}
\end{proposition}
\vspace{-0.3cm}
\begin{proof}
The DMT outer bounds presented in \eqref{eq:cut_dmt} are obtained by applying the cut-set bound \cite[Thm. 15.10.1]{cover-thomas:it-book} to the ICR. A detailed proof is provided in Appendix \ref{app:CUTSET}.
\end{proof}

We note that the mutual information expressions in the cut-set outer bound, which is the basis for deriving the DMT outer bound, are simultaneously maximized by mutually independent complex Normal channel inputs, each with zero-mean and unit power \cite[Appendix C]{Dabora:12}.

\vspace{-0.5 cm}

\subsection{An Achievable DMT Region Via Compress-and-Forward Relaying}
\label{sec:CFCOMMON}
\!We now derive an achievable DMT region based on CF relaying and  provide sufficient conditions under which this region coincides with the DMT outer bound \eqref{eq:cut_upper}, leading to the characterization of the {\em optimal} DMT of the ICR.\! The achievability scheme is based on a special case of \cite{Tian:11} in which the sources transmit only common messages.
\!The result is summarized in the following theorem:

\vspace{-0.4 cm}

\begin{theorem}
\label{thm:theorem_cf}
An inner bound on the DMT region of the symmetric slow-fading Gaussian ICR is given by:
\begin{equation}
    \label{eq:dmt_cf}d^{-}_{\textrm{CF}}(r_1,r_2) = \min_{k\in\{1, 2, 3\}}\left\{d_{k,\textrm{CF}}^{-}(r_1,r_2)\right\},
\end{equation}
where,
\begin{subequations}
\label{eq_cf_dmt}
\begin{eqnarray}
    \hspace{-0.5 cm}\label{eq:cf_dmt1}d_{1,\textrm{CF}}^{-}(r_1,r_2) &=&
    \begin{cases}
        (1\!-\!r_1)^+ + \left(\gamma\!-\!(\gamma\!+\!\alpha\!-\!\beta)^+\! -\!r_1\right)^+ & \qquad\qquad\qquad\qquad\qquad ,\alpha>1\\
        (1\!-\!r_1)^+ + \left(\gamma\!-\!(\gamma\!+\!1\!-\!\beta)^+ \!-\!r_1\right)^+   & \qquad\qquad\qquad\qquad\qquad ,\alpha\le1
    \end{cases}\\
    \label{eq:cf_dmt2}d_{2,\textrm{CF}}^{-}(r_1,r_2) &=&
    \begin{cases}
        (1\!-\!r_2)^+ + \left(\gamma\!-\!(\gamma\!+\!\alpha\!-\!\beta)^+ \!-\!r_2\right)^+ & \qquad\qquad\qquad\qquad\qquad\: ,\alpha>1\\
        (1\!-\!r_2)^+ + \left(\gamma\!-\!(\gamma\!+\!1\!-\!\beta)^+ \!-\!r_2\right)^+   & \qquad\qquad\qquad\qquad\qquad\: ,\alpha\le1
    \end{cases}\\
    \label{eq:cf_dmt3}d_{3,\textrm{CF}}^{-}(r_1,r_2) &=&
    \begin{cases}
        (1\!-\!r_1\!-\!r_2)^+ + (\alpha\!-\!r_1\!-\!r_2)^+ +\big(\gamma\!-\!(\gamma\!+\!\alpha\!-\!\beta)^+ \!-\!r_1\!-\!r_2\big)^+ &,\alpha>1\\
        (1\!-\!r_1\!-\!r_2)^+ + (\alpha\!-\!r_1\!-\!r_2)^+ +\big(\gamma\!-\!(\gamma\!+\!1\!-\!\beta)^+ \!-\!r_1\!-\!r_2\big)^+ &,\alpha\le1
    \end{cases}
\end{eqnarray}
\end{subequations}
\end{theorem}
\begin{proof}
The proof is provided in Appendix \ref{app:CFProof}.
\end{proof}
\vspace{-0.3 cm}
\begin{corollary}
\label{cor:cf_optimallity}
Consider the symmetric slow-fading Gaussian ICR defined in Section \ref{sec:Model}. If
    \begin{subequations}
    \label{eq:CFOpCon000}
    \begin{eqnarray}
        \label{eq:CFOpCon0}
        &&\beta\geq \max\{\gamma+1,\gamma+\alpha\}\\
        &&\min\big\{(1-r_1)^++(\gamma-r_1)^+,(1-r_2)^++(\gamma-r_2)^+\big\}\leq \nonumber\\ && \qquad\qquad\qquad\qquad\qquad\qquad (1-r_1-r_2)^++(\alpha-r_1-r_2)^++(\gamma-r_1-r_2)^+,       \label{eq:CFOpCon}
    \end{eqnarray}
    \end{subequations}
then the optimal DMT region is
\begin{equation}
    \label{eq:cf_optimal_dmt}
    d_{\mbox{\footnotesize Opt-CF}}(r_1,r_2) = \min\big\{(1-r_1)^++(\gamma-r_1)^+,(1-r_2)^++(\gamma-r_2)^+\big\},
\end{equation}
and it is achieved with CF at the relay.
\end{corollary}
\vspace{-0.3 cm}
\begin{proof}
    The corollary follows from Theorem \ref{thm:theorem_cf}. Note that when $ \beta \geq \max\{\gamma+1,\gamma+\alpha\}$ and \eqref{eq:CFOpCon} is satisfied, then the achievable DMT region of the CF scheme is $\min\big\{(1-r_1)^++(\gamma-r_1)^+,(1-r_2)^++(\gamma-r_2)^+\big\}$ which coincides with the DMT outer bound derived in Proposition \ref{thm:theorem_cutset}.
\end{proof}

\subsection{Discussion}
\label{com:discom1}

{\em The physical meaning of conditions \eqref{eq:CFOpCon000}:} To understand the physical meaning of conditions \eqref{eq:CFOpCon000} note that $\frac{\SNR^\beta}{\SNR}$ represents the relay-destination SNR, treating only the desired signal as noise, and $\frac{\SNR^\beta}{\SNR^\alpha}$ represents the relay-destination SNR, treating only the interference as noise. Condition \eqref{eq:CFOpCon0} can now be rewritten as $\min\{\beta-1, \beta-\alpha\}>\gamma$, or equivalently, $\min \left\{\frac{\SNR^\beta}{\SNR},\frac{\SNR^\beta}{\SNR^\alpha}\right\}> \SNR^\gamma$. Thus, \eqref{eq:CFOpCon0} can be interpreted as requiring that the SNR of the relay-destination link (i.e., $\min \left\{\frac{\SNR^\beta}{\SNR},\frac{\SNR^\beta}{\SNR^\alpha}\right\})$ be higher than the SNR of the source-relay link (i.e., $\SNR^{\gamma}$). It follows that this conditions guarantees that the relay be able to reliably convey its received information to destination nodes using only minor compression. The second condition, stated in \eqref{eq:CFOpCon}, requires that the multiplexing gains be such that jointly decoding both messages at the destination (with the help of the relay) does not constrain the individual rates. The combination of this condition with the fact that the relay can reliably convey its information to the destination nodes leads naturally to the optimality of the CF scheme used in Theorem \ref{eq_cf_dmt}.

    {\em On the equivalence to the DMT of two parallel relay channels:} Note that under the conditions of Corollary \ref{cor:cf_optimallity}, \eqref{eq:cf_optimal_dmt} corresponds to the optimal DMT of two interference-free parallel relay channels. This can be seen by inspecting the cut-set bound for the relay channel. The relationship between the channel inputs and outputs for the relay channel at times  $i=1,2,...,n$, is given by:
    \begin{eqnarray*}
        Y_{1,i} &=& \sqrt{\SNR}H_{11}X_{1,i}+\sqrt{\SNR^{\beta}}H_{31}X_{3,i}+Z_{1,i}\\
        Y_{3,i} &=& \sqrt{\SNR^{\gamma}}H_{13}X_{1,i}+Z_{3,i},
    \end{eqnarray*}
    where $X_1$ and $X_3$ denote the transmitted signals of the source and of the relay, respectively, and $Y_1$ and $Y_3$ denote the received signals at the destination and at the relay node, respectively. The rest of the definitions are as given in Section \ref{sec:Model}.
    Let $\uH_R = (H_{11},H_{13}, H_{31})$. From \cite[Eqns. (4), (5)]{Yuksel:07}, for a given realization $\uh_r$
    the capacity of the relay channel is upper-bounded by
    \begin{equation*}
        C_{\mbox{\footnotesize Relay}}(\uh_r)\!\!\le\!\!\max_{f(x_1,x_2)} \!\!\min\big\{I(X_1,X_3;Y_1|\uh_r), I(X_1;Y_1,Y_3|X_3,\uh_r)\big\}.
    \end{equation*}
    \noindent
    Following steps similar to \cite[Proof of Thm. 1]{Maric:11}, we bound each expression as follows: 
    \begin{eqnarray}
        I(X_1;Y_1,Y_3|X_3,\uh_r)&\le&\log\Big(1+\SNR|h_{11}|^2+\SNR^{\gamma}|h_{13}|^2\Big)\nonumber\\
        \label{eq:eq16}&=&\log\Big(1+\SNR^{1-\theta_{11}}+\SNR^{\gamma-\theta_{13}}\Big).\\
        \label{eq:eq15}I(X_1,X_3;Y_1|\uh_r)&\le&
        \log\big(\SNR^{1-\theta_{11}}+2\cdot\SNR^{\frac{1-\theta_{11}+\beta-\theta_{31}}{2}}+\SNR^{\beta-\theta_{31}}+1\big).
    \end{eqnarray}
    \noindent
    Hence, the DMT of the relay channel is upper-bounded by
    {\begin{equation}
        \label{eq:relayDMTUB}
       d_{\mbox{\footnotesize Relay}}^+(r)=\min\big\{(1-r)^++(\gamma-r)^+, (1-r)^++(\beta-r)\big\}.
    \end{equation}}
    Next, recall that in Corollary \ref{cor:cf_optimallity} we have $\beta\ge \gamma+1$. Comparing $d_{\mbox{\footnotesize Relay}}^+(r)$ and $d_{\mbox{\footnotesize Opt-CF}}(r,r)$, we conclude that for $\beta\ge \gamma+1$, $d_{\mbox{\footnotesize Relay}}^+(r)  =d_{\mbox{\footnotesize Opt-CF}}(r, r)$. Thus, under the conditions of Corollary \ref{cor:cf_optimallity}, the optimal DMT of the ICR {\em coincides with the outer bound
    on the optimal DMT of two interference-free parallel relay channels}. Hence, {\em a single relay} employing the CF strategy in this situation is DMT optimal for {\em both} communicating pairs simultaneously, and in fact interference does not degrade the DMT performance in this case.

    {\em Notes on the optimality of CF:} Observe that there exist values of $\alpha$ and $\beta$ for which DMT optimality of CF holds over the entire range of multiplexing gains, i.e., for any $0\!\le\! r_1,r_2\! \le\!1$. One such example is
    when $\gamma=1, \alpha\!=\!2$ and $\beta\!=\!3$. However, if $\beta\!<\!\max\{\gamma+1,\gamma+\alpha\}$,  CF is not DMT-optimal in the sense that its achievable DMT region does not coincide with the DMT outer bound derived based on the cut-set theorem (we cannot rule out a tighter outer bound for this scenario). To understand the reason for this, note that as stated in the previous discussion, when the relay-destination links are strong (large $\beta$), then the compression loss at the relay is minor which allows conveying the relay information to the destinations with negligible distortion. However, if the relay-destination links are not strong enough, i.e., if $\beta\!<\!\max\{\gamma+1,\gamma+\alpha\}$, then the compression loss at the relay is substantial. As a result, much of the information received at the relay is not conveyed to the destinations, leading to sub-optimality of CF. It follows that the achievable DMT region of the CF-based scheme used for deriving Theorem \ref{thm:theorem_cf} coincides with the DMT region derived from the cut-set bound only over a strict subset of channel coefficients, {\em contrary} to the situation for the CF scheme in the single-relay channel, as observed in \cite{Yuksel:07}. This difference is due to the fact that unlike the single-relay channel, in the ICR there is interference, and hence the relay-destination links should be stronger than in the case of single-relay channel in order for CF to be optimal. This follows since the destination has to be able to decode the relay signal in the presence of interference. An alternative to CF in situations where the relay-destination links are too weak for CF optimality is to use the DF scheme. However, DF has other limitations due to the required source-relay link strength. This is discussed in detail in the next section.

    Furthermore, note that from Theorem \ref{thm:theorem_cf} it follows that when $\alpha\le1$, i.e., when interference is weak, the achievable DMT region obtained with CF at the relay (when each receiver decodes both messages) is an increasing function of $\alpha$, that is, increasing the interference between the communicating pairs enlarges the DMT region. On the other hand, if $\alpha>1$, i.e., if interference is strong, then there are two cases:
    \begin{itemize}
        \item When $\beta\geq\gamma+\alpha$, then \eqref{eq:cf_dmt1} and \eqref{eq:cf_dmt2} do not depend on $\alpha$,  while \eqref{eq:cf_dmt3} increases as $\alpha$ increases.
        We conclude that in this case, the DMT region with CF is enlarged as the interference  becomes stronger. This is because the cooperation links are strong enough to support the required relay information rate and  to facilitate decoding of the interference and of the desired message at each receiver.

        \item  When $\beta<\gamma+\alpha$ and $r_1+r_2\leq\min\{\beta-\alpha,\alpha\}$, then 
          \eqref{eq:cf_dmt1} and \eqref{eq:cf_dmt2} decrease as  $\alpha$ increases, while \eqref{eq:cf_dmt3} does not depend on $\alpha$. Thus, increasing the interference decreases the achievable DMT obtained with the CF strategy. This follows as the SNR at the relay-destination links is not high enough for the relay to convey the information on its received signal to the destinations without substantial compression, and hence, its assistance is limited. Thus, increasing $\alpha$ decreases the DMT.
    \end{itemize}
    Additionally, it can be observed that the achievable DMT region of CF is  a non-decreasing function of $\beta$, which represents the strength of the cooperation links from the relay to the destinations. Thus, better relay-destinations links (i.e.,  better cooperation) improves the DMT performance for the ICR. Similarly, note that the DMT region of CF is a non-decreasing function of $\gamma$, which corresponds to the strength of the source-relay links. Thus, strong source-relay links (i.e., increasing the reliability of the information at the relay) increases the achievable DMT for the ICR as well.

    {\em Maximum achievable diversity gain with CF:} For the Gaussian ICR defined in Section~\ref{sec:Model}, the maximum achievable diversity gain with CF relaying is:
    \begin{equation*}
       \qquad \qquad \qquad d^{\max}_{\mbox{\footnotesize CF}} =
        \begin{cases}
            1+\min\big\{\gamma,(\beta-\alpha)^+\big\} &  \qquad \qquad ,\alpha>1\\
            1+\min\big\{\gamma,(\beta-1)^+\big\} &  \qquad \qquad ,\alpha\le 1.
        \end{cases}
    \end{equation*}
    In \cite{Yuksel:07} it is shown that for the single relay channel with $\gamma=\beta=1$, the maximum achievable diversity gain is $2$ and it is achieved with CF at the relay. Note that in the ICR with $\gamma=1$, and $\alpha\le1$ there are two cases:
     \begin{itemize}
        \item If $\beta\geq 2$, then we have $d^{\max}_{\mbox{\footnotesize CF}}=2$
        \item If $\beta<2$, then we have $d^{\max}_{\mbox{\footnotesize CF}}=\beta<2.$
    \end{itemize}
    Hence, when $\alpha\le1$, although the DMT region of the ICR increases with $\alpha$ (see also the discussion above), the {\em maximal diversity gain} of the ICR does not depend on $\alpha$. Lastly, we note that when $\alpha\le 1$, $\beta<2$, and $\gamma=1$, then the maximal diversity gain of the CF scheme (as implemented in Theorem \ref{thm:theorem_cf}) is smaller than the diversity gain of the relay channel. Thus, interference decreases the diversity gain (subject to joint decoding at the receivers).

    \noindent
    For the scenario with $\gamma=1$ and $\alpha>1$:
     \begin{itemize}
        \item If $\beta\geq 1+\alpha$, then we have $d^{\max}_{\mbox{\footnotesize CF}}=2$
        \item If $\beta<1+\alpha$, then we have $d^{\max}_{\mbox{\footnotesize CF}}=1+(\beta-\alpha)^+\le2.$
    \end{itemize}
    It follows that when $\gamma=1$, $\alpha>1$ and $\beta\geq 1+\alpha$, although the DMT region of the ICR generally increases with $\alpha$, the maximal diversity gain of the ICR does not depend on $\alpha$. For the scenario in which $\gamma=1$, $\alpha>1$, and $\beta< 1+\alpha$, both the maximal diversity gain and the DMT region of the ICR decrease with respect to $\alpha$.

    In conclusion, in the ICR (with $\gamma=1$) the same diversity gain as that of the relay channel is achieved if $\beta\ge\max\{2,\alpha+1\}$. This is due to the fact that unlike the single-relay channel, in the ICR there is interference and therefore, in order to achieve the same diversity gain, the relay-destination links in the ICR should be stronger than that in the single-relay channel to overcome the interference.

    {\em CF relaying can increase the DMT of the IC:} We note that the relay can enlarge the DMT region of the IC in certain regimes:
    The DMT region of the IC without a relay was outer bounded in \cite[Thm. 1]{Leveque:09} 
    by $\min\big\{(1-r_1)^+,(1-r_2)^+\big\}$. This outer bound can be achieved in certain regimes, e.g., in the very strong interference regime, characterized by $\alpha\ge 2$
    \cite[Section VI]{Bolcskei:09}.
    Note that in the scenario considered in Corollary \ref{cor:cf_optimallity}, if we consider $\gamma=1$, then the achievable DMT of ICR is {\em twice the maximum achievable DMT of the IC}. Note that for values of $\gamma>1$, this gap becomes even larger, i.e., the achievable DMT of the ICR becomes strictly larger than twice the maximum achievable DMT of the IC. Furthermore, from Theorem \ref{thm:theorem_cf} it follows that the DMT performance of the ICR is better than that of  the IC also in the scenarios where CF is not DMT optimal:
    for example, when $\beta\ge\max\{\gamma+1,\gamma+\alpha\}$ and
    \begin{eqnarray*}
        &&\min\big\{(1-r_1)^+,(1-r_2)^+\big\}\leq (1-r_1-r_2)^++(\gamma-r_1-r_2)^++(\alpha-r_1-r_2)^+\le\\
        &&\qquad\qquad\qquad\qquad\qquad\qquad\qquad\qquad\qquad\min\big\{(1-r_1)^++(\gamma-r_1)^+,(1-r_2)^++(\gamma-r_2)^+\big\},
    \end{eqnarray*}
    the optimal DMT region of the IC \cite{Bolcskei:09} is $d_{\mbox{\footnotesize Opt-IC}}(r_1,r_2)=\min\big\{(1-r_1)^+,(1-r_2)^+\big\}$, while for the ICR $d_{\mbox{\footnotesize CF}}^-(r_1,r_2)= (1-r_1-r_2)^++(\gamma-r_1-r_2)^++(\alpha-r_1-r_2)^+$, and hence, $d_{\mbox{\footnotesize Opt-IC}}(r_1,r_2)\le d_{\mbox{\footnotesize CF}}^-(r_1,r_2)$. These observations motivate the use of relaying in wireless networks.





\section{Rx-CSI Only at the Relay: Decode-and-Forward}
\label{sec:DMTDF}

\vspace{-0.2 cm}

\subsection{DMT Outer Bound}
\label{sec:OBDF}

\vspace{-0.2 cm}

Since the case with only Rx-CSI at the relay is a special case of Proposition \ref{thm:theorem_cutset},  the outer bound of Proposition \ref{thm:theorem_cutset} is an outer bound also for the  case where both the receivers and the relay have only Rx-CSI. Thus, we have the following proposition:
\vspace{-0.4 cm}
\begin{proposition}
    For the symmetric ICR with Rx-CSI only at the relay, 
    the region $d^+(r_1,r_2)$ defined in Equations \eqref{eq:cut_upper}-\eqref{eq:cut_dmt}
    is an outer bound on the DMT region.
\end{proposition}


\vspace{-0.5 cm}

\subsection{An Achievable DMT Region Via Decode-and-Forward Relaying}
\label{sec:DFCOMMON}
\vspace{-0.2 cm}
We begin with a statement of the achievable DMT region:

\vspace{-0.4 cm}

\begin{theorem}
\label{thm:theorem_df}
Define $d^{-}_{\textrm{DF}}(r_1,r_2) $ as follows:
\begin{equation}
\label{eq:df_achievable_dmt}
d^{-}_{\textrm{DF}}(r_1,r_2) =
\begin{cases}
    \min\big\{d_{\mbox{\tiny IC}}(r_1,r_2)+d_{\mbox{\tiny Relay}}(r_1,r_2),d_{\mbox{\tiny Coop.}}(r_1,r_2)\big\} & \quad r_1+r_2<\gamma\\
    d_{\mbox{\tiny IC}}(r_1,r_2) & \quad r_1+r_2\ge\gamma,
\end{cases}
\end{equation}
\vspace{-0.4 cm}
where
\begin{subequations}
\begin{eqnarray}
    \label{eqn:d_RELAYDF}
    d_{\mbox{\tiny Relay}}(r_1,r_2)&=&\min\Big\{(\gamma-r_1)^+,(\gamma-r_2)^+,2(\gamma-r_1-r_2)^+\Big\}\\
    \label{eqn:d_IC}
    d_{\mbox{\tiny IC}}(r_1,r_2)&=&\min\Big\{(1-r_1)^+,(1-r_2)^+,(1-r_1-r_2)^+ + (\alpha-r_1-r_2)^+\Big\}\\
    d_{\mbox{\tiny Coop.}}(r_1,r_2)&=& \min\Big\{(1-r_1)^++(\beta-r_1)^+,(1-r_2)^++(\beta-r_2)^+,\nonumber\\
    &&\qquad\qquad(1-r_1-r_2)^++(\alpha-r_1-r_2)^++(\beta-r_1-r_2)^+\Big\}.
\end{eqnarray}
\end{subequations}
Then, $d^{-}_{\textrm{DF}}(r_1,r_2)$ is an achievable DMT for the symmetric slow-fading Gaussian ICR.
\end{theorem}
\begin{proof}
The proof is provided in Appendix \ref{app:DFProof}.
\end{proof}

\begin{corollary}
\label{cor:df_optimal_con}
Consider the symmetric slow-fading Gaussian ICR as defined in Section \ref{sec:Model}.
If the following conditions are satisfied:
\begin{subequations}
\label{eq:df_optimal_con}
\begin{eqnarray}
    \label{eq:df_optimal_con2}&&r_1+r_2\leq \gamma\\
    \label{eq:df_optimal_con2-2}&&\max\big\{(\gamma-r_1)^+,(\gamma-r_2)^+\big\}\leq 2(\gamma-r_1-r_2)^+\\
    \label{eq:df_optimal_con3}&&\max\big\{(1-r_1),(1-r_2)\big\}\leq (1-r_1-r_2)^++(\alpha-r_1-r_2)^+\\
    \label{eq:df_optimal_con4}&&\max\big\{(1-r_1)+(\beta-r_1)^+,(1-r_2)+(\beta-r_2)^+\big\}\leq (1-r_1-r_2)^+ \nonumber\\ &&
    \qquad\qquad\qquad\qquad\qquad\qquad\qquad\qquad\qquad\qquad+(\alpha-r_1-r_2)^++(\beta-r_1-r_2)^+,
\end{eqnarray}
\end{subequations}
then the optimal DMT is
\vspace{-0.2 cm}
\begin{eqnarray}
    d_{\footnotesize Opt-DF}(r_1,r_2) &=& \min\big\{(1-r_1)^++(\gamma-r_1)^+,(1-r_2)^++(\gamma-r_2)^+,\nonumber\\
    &&\label{eq:df_optimal_dmt}\qquad\qquad(1-r_1)^++(\beta-r_1)^+,(1-r_2)^++(\beta-r_2)^+\big\},\vspace{-0.3 cm}
\end{eqnarray}
and it is achieved with DF at the relay.
\end{corollary}
\begin{proof}
The proof is based on Theorem \ref{thm:theorem_df}. Note that when \eqref{eq:df_optimal_con} is satisfied, then \eqref{eq:df_achievable_dmt} coincides with \eqref{eq:cut_upper}, characterizing the optimal DMT for the Gaussian ICR.
\end{proof}

\subsection{Discussion}
\label{com:discom2}

{\em The physical meaning of conditions \eqref{eq:df_optimal_con}:} There are four conditions in Corollary \ref{cor:df_optimal_con}. These conditions correspond to different requirements on the multiplexing gains and on the SNR exponents in the channel. The conditions can be interpreted as follows:
\begin{itemize}
    \item The first condition in \eqref{eq:df_optimal_con2} requires that the multiplexing gains be small enough such that jointly decoding both messages at the relay does not result in an outage, otherwise, DF is not useful and the ICR scales back to the IC, see \cite{Bolcskei:09}.
    \item The second condition in \eqref{eq:df_optimal_con2-2} requires that joint decoding at the relay does not limit the individual rates at the relay.
    \item The third condition in \eqref{eq:df_optimal_con3} requires that when the relay cannot help (i.e., is in outage), then jointly decoding both the desired message and the interference at the destination does not constrain the rate of the desired message.
    \item The fourth condition in \eqref{eq:df_optimal_con4} requires that given a successful decoding at the relay node, jointly decoding both the desired message and the interference at the destination does not constrain the rate of the desired message.
\end{itemize}
We obtain that the sum of \eqref{eq:df_optimal_con2-2} and \eqref{eq:df_optimal_con3} corresponds to the outage probability at the destination when there is outage at the relay, and \eqref{eq:df_optimal_con4} corresponds to the outage probability at the destination when the relay is not in outage. Thus, \eqref{eq:df_optimal_con2-2} combined with \eqref{eq:df_optimal_con3} guarantee DMT optimality when the relay is in outage, and \eqref{eq:df_optimal_con4} guarantees DMT optimality when the relay is not in outage.

{\em On the optimality of DF:} From Theorem \ref{thm:theorem_df} we observe that the achievable DMT region obtained with DF (when each receiver decodes both messages) is monotonically non-decreasing as $\alpha$, $\beta$, and $\gamma$ increase, i.e., the DMT performance of the DF scheme improves as either interference, $\alpha$, or cooperation (either the relay-destination links, $\beta$, or the source-relay links, $\gamma$) become stronger, or it does not change when they increase. This is since increasing $\alpha$ facilitates joint decoding at the destinations. This is contrary to CF for which there are regimes of $\alpha$ and $\beta$ in which increasing the interference decreases the DMT performance.

    {\em On the equivalence to the DMT of two parallel relay channels:} Setting $r_1=r_2=r$, we note that under conditions \eqref{eq:df_optimal_con} each Tx-Rx pair achieves a DMT of $\min\big\{(1-r)+(\gamma-r_1)^+,(1-r)+(\beta-r)^+\big\}$, which coincides with the
     DMT upper bound for the relay channel \eqref{eq:relayDMTUB}. Thus, the optimal DMT \eqref{eq:df_optimal_dmt} in this case corresponds to the optimal DMT of two interference-free
     parallel relay channels. We conclude that the DF strategy can be {\em DMT optimal for both communicating pairs simultaneously}. Note that with CF, DMT optimality was shown only for $\beta \ge \max\{\gamma+1,\gamma+\alpha\}$; but with DF, optimality is achieved for any value of $\beta$ satisfying \eqref{eq:df_optimal_con}, which represents a wider range of values than that for CF. The DMT optimality of DF for different scenarios is demonstrated in Figures \ref{fig:plots123}-\ref{fig:plots891011} in the next section. For example, observe in Fig. \ref{fig:plot211} that for $\beta=1$ DF is DMT optimal for some multiplexing gains (e.g. $r_1=r_2\le \frac{1}{3}$) while CF is suboptimal for all values of multiplexing gains. 

{\em Maximum achievable diversity gain with DF:} The maximum diversity gain achieved by the DF scheme is $d^{\mbox{\scriptsize max}}_{\textrm{DF}}\!=\!\min\{\gamma+1, \beta+1\}$. Compared with the relay channel whose maximum diversity gain is $2$, we conclude that the~DF scheme achieves for each pair the {\em maximum diversity} gain of the relay channel, as long as $\min\{\beta,\gamma\}\geq 1$. Observe that this diversity gain is obtained for both pairs simultaneously, using only {\em a single}~relay.

{\em DF relaying can increase the DMT of the IC:} When DF is DMT optimal, its DMT region outer bounds the optimal DMT region of the IC (for the same $\alpha$). Note that this conclusion holds for any value of $\beta>0$. Moreover, there are scenarios in which the achievable DMT for the ICR with DF is strictly larger than the optimal DMT of the IC even when DF is not optimal. One such example is when \eqref{eq:df_optimal_con2} and \eqref{eq:df_optimal_con4} are satisfied, while \eqref{eq:df_optimal_con3} is not satisfied. For example, when $r_1=r_2=0.4$, $\alpha=1.8$, $\beta=1$, and $\gamma=1$  then for the ICR we have a diversity gain of $1$, while for the IC the achievable diversity gain is {\em upper-bounded} by $0.6$.

\section{Amplify-and-Forward at the Relay}
\label{sec:DMTAF}
In this section we study scenarios in which the relay node uses the AF scheme. We first consider a relay operating in the full-duplex mode and then study
relaying subject to a half-duplex constraint. For each mode we propose a transmission scheme and evaluate its DMT performance.
In addition to our standard Rx-CSI assumption, we assume that each receiver knows the Rx-CSI at the relay. This can be done by sending the Rx-CSI at the relay to the receivers with a negligible rate cost, as the channel is constant during the transmission of the entire codeword. For simplicity we consider only scenarios in which $\gamma=1$.

\subsection{An Outer Bound on the DMT Region}
\vspace{-0.2 cm}
Note that the DMT region $d^{+}(r_1,r_2)$ defined in \eqref{eq:cut_upper}-\eqref{eq:cut_dmt} was derived when each receiver has CSI on all links in the channel (see Appendix \ref{app:CUTSET} for details). Therefore, $d^+(r_1,r_2)$ is an outer bound on the DMT region of the ICR with a full-duplex relay, and the ICR with a half-duplex relay.

\subsection{An Achievable DMT Region Via Full-Duplex Amplify-and-Forward Relaying}
An achievable DMT region for the ICR with a full-duplex relay employing the AF strategy is stated in the following theorem:
\begin{theorem}
\label{thm:theorem_af}
Let $d^{-}_{\textrm{AF}_{FD}}(r_1,r_2)$ be defined as follows:
\begin{equation}
    \label{eq:af_achievable_dmt}
    d_{\textrm{AF}_{FD}}^{-}(r_1,r_2)=
    \begin{cases}
        \!\min\Big\{\!(1-r_1)^+,(1-r_2)^+,(\alpha-1-r_1)^+,(\alpha-1-r_2)^+\!\Big\}  &   \hspace{-1.3 cm}\! ,\beta<1\\
        \!\min\Big\{\!(2-\beta-r_1)^++(\beta-1-r_1)^+,(2-\beta-r_2)^++(\beta-1-r_2)^+,\\
        \hspace{3 cm}(\alpha-\beta-r_1)^+,(\alpha-\beta-r_2)^+\!\Big\}  & \hspace{-2.1 cm}\!  ,1\le\beta<2\\
        \!\min\Big\{\!(1-r_1)^+,(1-r_2)^+,(\alpha-\beta-r_1)^+,(\alpha-\beta-r_2)^+\!\Big\} & \! \hspace{-1.3 cm} ,\beta\ge2
    \end{cases}
\end{equation}
The DMT region $d^{-}_{\textrm{AF}_{FD}}(r_1,r_2)$ is achievable  for the symmetric slow-fading Gaussian ICR.
\end{theorem}

\begin{proof}
The proof is provided in Appendix \ref{AFFDProof}.
\end{proof}

\subsection{An Achievable DMT Region Via Half-Duplex Amplify-and-Forward Relaying}
The noise amplification issue observed when the relay employs FD-AF at the relay motivates the consideration of half-duplex relay operation, with the goal of limiting the noise amplification,
 and thereby potentially increasing the diversity gain.
 In this section, we consider HD-AF relaying in which each receiver jointly decodes both its desired message and the interfering message. The corresponding DMT region, denoted by $d^{-}_{\textrm{AF}_{HD}}(r_1,r_2)$, is given in the following theorem:

\vspace{-0.4cm}
\begin{theorem}
\label{thm:theoremafjd}
Let $d^{-}_{\textrm{AF}_{HD}}(r_1,r_2)$ be defined as follows:
\begin{eqnarray}
\label{eq:HDAFDMT1}
\!\!\!\!\!\!d^{-}_{\textrm{AF}_{HD}}(r_1,r_2) \!= \!\begin{cases}
\!\min\!\bigg\{\!(1-r_1)^+\! +(\beta-2r_1)^+,(1-r_2)^++(\beta-2r_2)^+,&\\
\qquad\quad (1-r_1-r_2)^+ \!+(\alpha-r_1-r_2)^++(\beta-2r_1-2r_2)^+\!\bigg\}  &\!\!\! ,\beta\!\le\! 1\\
\!\min\!\bigg\{\!\max\left\{2(1-2r_1)^+,(1-2r_1)^++(\frac{3-\beta}{2}-r_1)^+\right\},&\\
    \quad\quad\max\left\{2(1-2r_2)^+,(1-2r_2)^++(\frac{3-\beta}{2}-r_2)^+\right\},&\\
\quad\quad(\frac{3-\beta}{2}-r_1-r_2)^+\!+\!(\frac{2\alpha+1-\beta}{2}-r_1-r_2)^+\!+\!(1-2r_1-2r_2)^+ \!\bigg\}  &\!\!\! ,\beta\!> \! 1
\end{cases}
\end{eqnarray}
The DMT region $d^{-}_{\textrm{AF}_{HD,JD}}(r_1,r_2)$ is achievable for the symmetric slow-fading Gaussian ICR.
\end{theorem}

\vspace{-0.3 cm}

\begin{proof}
The proof is provided in Appendix \ref{app:AFHDProof}.
\end{proof}

\subsection{Discussion}
    {\em On the equivalence to the DMT of two parallel relay channels:} From Theorem \ref{thm:theoremafjd} it follows that when $\beta=1$ and $\alpha\ge2$, i.e., when interference is very strong, then the achievable DMT is
    \begin{eqnarray*}
         d^{-}_{\textrm{AF}_{HD}}(r_1,r_2) = \min\Big\{(1-r_1)^++(1-2r_1)^+,(1-r_2)^++(1-2r_2)^+\Big\} .
    \end{eqnarray*}
    In \cite{Azarian:05} the class of AF relay channels has been studied and it has been shown that the nonorthogonal AF (NAF) protocol achieves the optimal DMT for AF single-relay channels which was shown in \cite[Theorem 3]{Azarian:05} to be $d^*(r)=(1-r)^++(1-2r)^+$. Thus, if $\beta=1$ and $\alpha\ge2$, then in the class of AF protocols, {\em our proposed HD-AF scheme achieves the optimal DMT for each communicating pair simultaneously, and the DMT performance corresponds to that of two parallel relay channels}. In fact, in this configuration, interference does not degrade the performance. This is because here interference is very strong, and thus, decoding the interfering message can be done without constraining the rate of the desired information.

    {\em The impact of noise amplification on the achievable DMT of AF:} Observe that for AF with $\beta\le1$, the achievable DMT of the ICR increases with respect to $\beta$, while for $\beta>1$, the DMT decreases with respect to $\beta$. This behaviour for AF can  be observed both in the strong interference regime (Figure \ref{fig:plots4567}) as well as in the weak to moderate interference regime (Figure \ref{fig:plots891011}). Hence, if the relay-destination links are weak, then forwarding desired information dominates the noise amplification caused by AF, while for strong relay-destination links, we observe the opposite behaviour. This demonstrates well the tradeoff between forwarding desired information to the receivers and amplifying the noise at the receivers.

\section{Numerical Evaluations and Additional Comments}

\begin{figure}[!t]
\centering
    \hspace{-0.95 cm}
    \subfloat[$\alpha=0.5,\mbox{ } \beta=\gamma=1$\label{fig:plot051}]{%
      \includegraphics[scale=0.41]{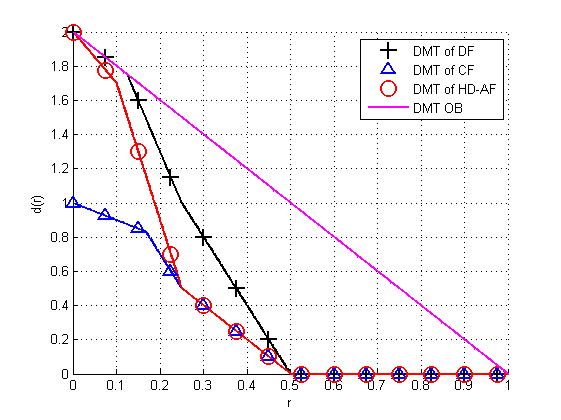}
    }
    \hspace{-0.85 cm}
    \subfloat[$\alpha=1,\mbox{ } \beta=\gamma=1$\label{fig:plot11}]{%
      \includegraphics[scale=0.41]{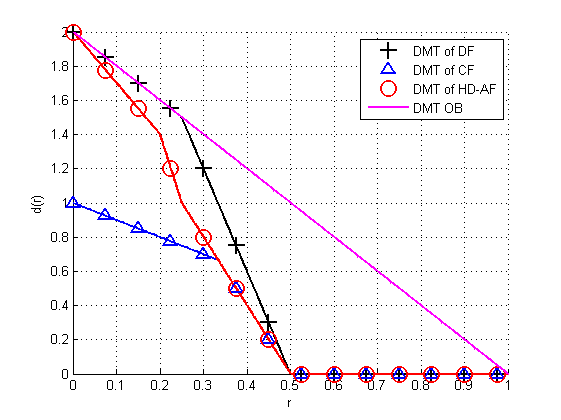}
    }
    \hspace{-0.85 cm}
    \subfloat[$\alpha\ge2,\mbox{ } \beta=\gamma=1$\label{fig:plot211}]{%
      \includegraphics[scale=0.41]{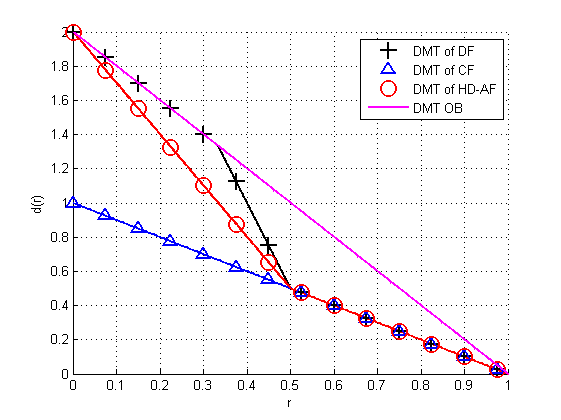}
    }
    \hspace{-1 cm}
    \caption{\small The effect of the strength of the interference on the achievable DMT of the ICR.}
    \label{fig:plots123}
\end{figure}

    {\em The effect of the strength of the interference on the achievable DMT:} Figure \ref{fig:plots123} depicts the achievable DMT of the ICR for different values of $\alpha$ when $\beta=1$, for the symmetric case where $r_1=r_2=r$. The figure demonstrates the effect of the strength of the interference on the achievable DMT. When $\alpha\le1$, i.e., when interference is weak, as in Figures \ref{fig:plot051} and \ref{fig:plot11}, the achievable diversity gain of HD-AF (Theorem \ref{thm:theoremafjd}), CF (Theorem \ref{thm:theorem_cf}), and DF (Theorem \ref{thm:theorem_df}) is equal to zero for high multiplexing gains ($r\ge0.5$).
    For CF and AF, this is due to fact that interference is not strong enough to facilitate joint decoding of the interference and of the desired message at the destinations (see equations \eqref{eq:cf_dmt3} for CF and \eqref{eq:HDAFDMT1} for AF). For DF, this is due to jointly decoding the messages from both sources at the relay node, which follows from \eqref{eqn:d_RELAYDF}. 
    However, when $\alpha=2$, i.e., when interference is very strong, decoding the interference at the destinations does not constrain the achievable DMT for high multiplexing gains. In this case, an outage for decoding the desired message at the destinations is the dominant outage event. This is the situation in Figure \ref{fig:plot211}.

    {\em The effect of the strength of the relay-destination links in the very strong interference regime:} Figure \ref{fig:plots4567} demonstrates the effect of the strength of the relay-destination links, represented by $\beta$, on the achievable DMT region of the ICR for the symmetric case where $r_1=r_2=r$ when interference is very strong ($\alpha = 2$). This makes it possible to isolate the effect of the relay-destination links. Note that when $\beta$ is small, as in Figure \ref{fig:plot202}, then the different relaying strategies achieve the same diversity gains for almost all values of multiplexing gains. This observation suggests that if the relay-destination link is very poor, then it does not matter which relaying strategy is used since the relay cannot provide much assistance to the communicating pairs. In fact, the DMT of the ICR in this case coincides with the DMT of the IC except for very low multiplexing gains, in which DF and AF provide DMT gain over the IC but CF does not.
    However, for $\beta=1$ and $\beta=2$ (Figures \ref{fig:plot21} and \ref{fig:plot22}, respectively), the achievable DMT of DF and AF reaches the maximum possible diversity gain at $r=0$, i.e., a diversity gain of $2$.

    Recall that in Corollary \ref{cor:cf_optimallity} it was shown that when $\beta\ge\max\{2,\alpha+1\}$, CF can be DMT-optimal. Indeed, in Figures \ref{fig:plot202}-\ref{fig:plot22}, where we have $\beta<\alpha$, CF is suboptimal and its achievable DMT is bounded by $1$. But, when $\beta\ge\max\{2,\alpha+1\}$, then CF becomes DMT-optimal, as is the situation in Figure \ref{fig:plot23}. Observe that for $\beta=2,3$
     (Figures \ref{fig:plot22}, \ref{fig:plot23}) the DMT of  AF decreases, and in fact becomes zero at high multiplexing gains. This is because when the relay-destination link is strong, then the noise amplification problem associated with AF becomes dominant and constrains the achievable DMT at the destinations.

\begin{figure}[!t]
\centering
    \hspace{-0.8 cm}
    \subfloat[$\alpha=2,\mbox{ } \beta=0.2,\mbox{ } \gamma=1$\label{fig:plot202}]{%
      \includegraphics[scale=0.31]{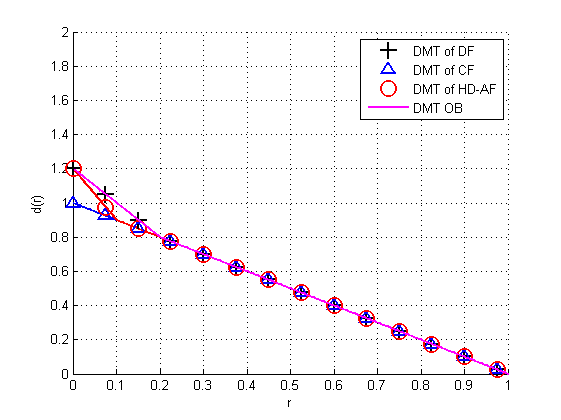}
    }
    \hspace{-0.72 cm}
    \subfloat[$\alpha=2,\mbox{ } \beta=1,\mbox{ } \gamma=1$\label{fig:plot21}]{%
      \includegraphics[scale=0.31]{21.png}
    }
    \hspace{-0.72 cm}
    \subfloat[$\alpha=2,\mbox{ } \beta=2,\mbox{ } \gamma=1$\label{fig:plot22}]{%
      \includegraphics[scale=0.31]{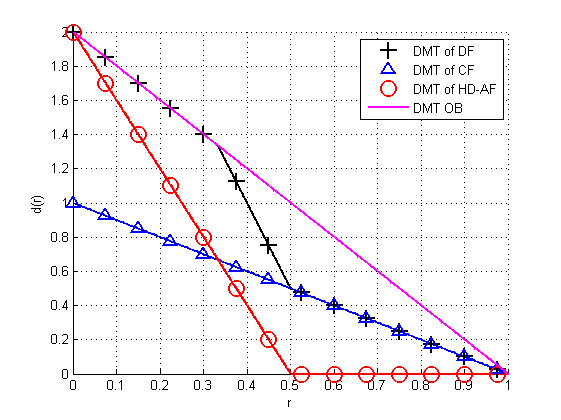}
    }
    \hspace{-0.72 cm}
    \subfloat[$\alpha=2,\mbox{ } \beta=3,\mbox{ } \gamma=1$\label{fig:plot23}]{%
      \includegraphics[scale=0.31]{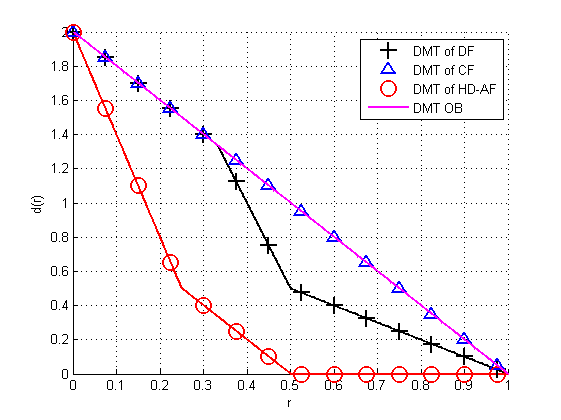}
    }
    \vspace{-0.4cm}
    \hspace{-1.05 cm}
    \caption{\small The effect of the strength of the relay-destination links on the achievable DMT of the ICR in the strong and in the very strong interference regime.}
    \label{fig:plots4567}
\end{figure}

    {\em The effect of the strength of the relay-destination links in the weak interference regime:} Figure \ref{fig:plots891011} demonstrates the effect of the strength of the relay-destination links on the achievable DMT region of the ICR in scenarios in which the interference is weak ($\alpha=0.5$). First, observe that in the weak interference regime, DF outperforms both CF and AF. Note that if the multiplexing gains are high ($r\ge0.5$), then the achievable DMT of all three strategies is equal to zero. In this case, the outage event due to jointly decoding both messages at the relay is the dominating outage event for DF, while for the CF and the AF relaying strategies, the dominating outage event is the one that corresponds to jointly decoding both messages at the destinations (See also the comment on the effect of the strength of the interference on the achievable DMT). When the multiplexing gains are low, however, then with the DF strategy, the relay can reliably decode both messages and forward {\em noiseless} desired information to the receivers, while with CF and AF strategies, the relay forwards noisy information to the receivers. Thus, DF outperforms CF and AF at low multiplexing gains.

    \noindent
    Note that the performance of CF improves with respect to the strength of the relay-destination links, i.e., when $\beta$ increases (note that for $\beta=0.5$ CF performs the same as for $\beta=1$, but when $\beta>1$, the DMT performance of CF improves as $\beta$ increases). This follows from the fact that as the relay-destination links improve, then the compression at the relay can be less substantial (see Eqn. \eqref{eq:eq25New}), enabling the relay to forward more information to the destinations.

\begin{figure}[!t]
\centering
    \hspace{-0.8 cm}
    \subfloat[$\alpha=0.5,\mbox{ } \beta=0.5,\mbox{ } \gamma=1$\label{fig:plot1a05b05}]{%
      \includegraphics[scale=0.31]{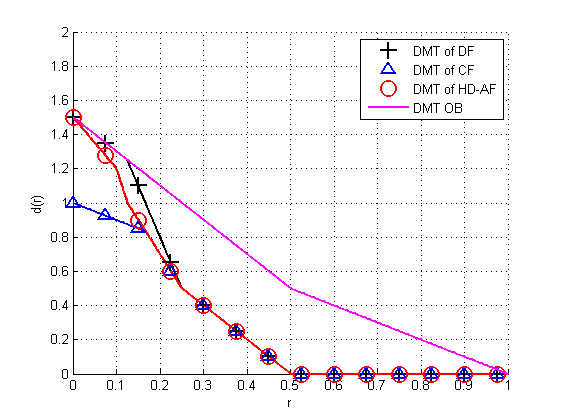}
    }
    \hspace{-0.72 cm}
    \subfloat[$\alpha=0.5,\mbox{ } \beta=1,\mbox{ } \gamma=1$\label{fig:plot2a05b1}]{%
      \includegraphics[scale=0.31]{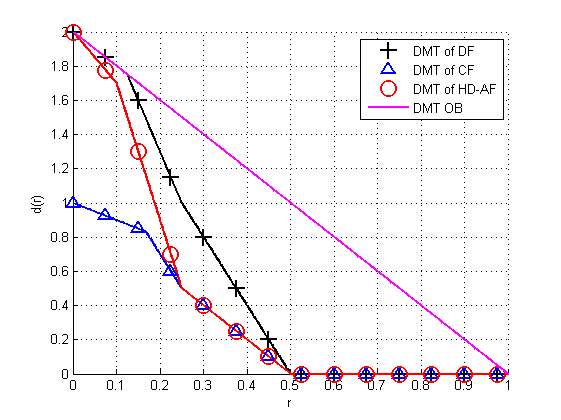}
    }
    \hspace{-0.72 cm}
    \subfloat[$\alpha=0.5,\mbox{ } \beta=1.5,\mbox{ } \gamma=1$\label{fig:plot3a05b15}]{%
      \includegraphics[scale=0.31]{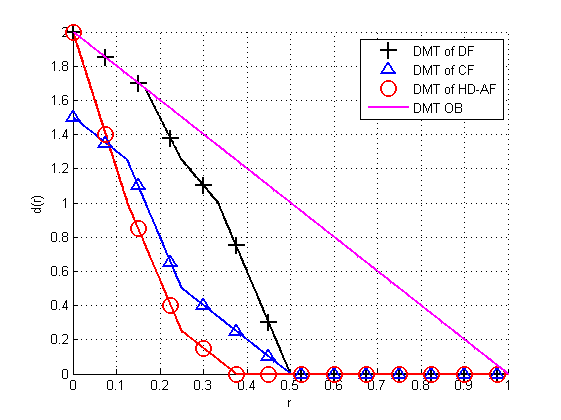}
    }
    \hspace{-0.72 cm}
    \subfloat[$\alpha=0.5,\mbox{ } \beta\ge3,\mbox{ } \gamma=1$\label{fig:plot4a05b3}]{%
      \includegraphics[scale=0.31]{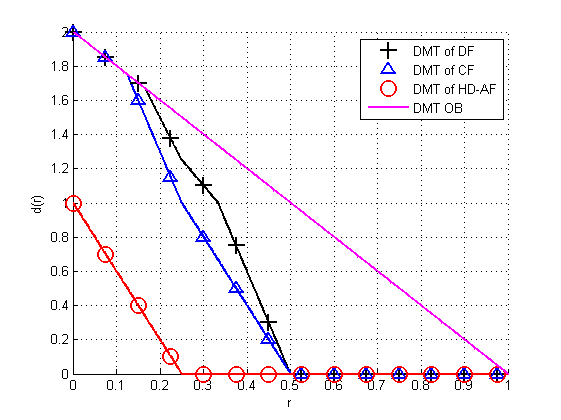}
    }
    \vspace{-0.4cm}
    \hspace{-1.05 cm}
    \caption{\small The effect of the strength of the relay-destination links on the achievable DMT of the ICR in the weak interference regime.}
    \label{fig:plots891011}
\end{figure}

    {\em Implications on the incorporation of relaying into existing wireless networks:} An important aspect to note is that all achievable DMT regions in this paper were obtained with mutually independent codebooks. This means that when attempting to achieve the DMT gains characterized in this work by adding a relay to an existing network, {\em it is not required to change the transmission scheme of the users}, and in fact they can be completely oblivious to the fact that they are being assisted by a relay node. It is enough to modify only the decoding process at the receivers. This greatly simplifies the introduction of relaying into wireless networks and provides further motivation for using relaying to mitigate interference.

    {\em Alternative approached for the weak interference regime:} The focus of the optimal DMT results obtained in this work is on the strong and on the very strong interference regimes. In these regimes, decoding both the desired message and the interfering message at each receiver achieves the optimal performance, since when interference is strong enough, it can be decoded without constraining the rate of the desired information. This observation has motivated basing the achievable schemes employed in this work on decoding both messages at each receiver. In the weak interference regime, this decoding approach constrains the rates, and higher rates can be obtained by applying partial interference cancellation as in the well-known Han-Kobayashi scheme \cite{Han:81}. Partial interference cancellation can be incorporated into these schemes by rate-splitting at the sources combined with partial decoding and rate-splitting at the relay (for example, this was done for CF in \cite[Sec. III-A]{Tian:11}). Note that for the IC (without a relay) it was shown in \cite{Bolcskei:09} that when interference is weak, then partial interference cancellation leads to a larger achievable DMT region compared to jointly decoding both messages; however, DMT optimality was demonstrated only for the strong and for the very strong interference regimes.

    Another relevant relaying strategy is noisy network coding (NNC) \cite{Lim:11}. Note that in the weak interference regime, \cite{Lim:11} showed that NNC may allow higher rate pairs than those achievable with CF relaying for the ICR with noiseless, orthogonal relay-destinations links. Thus, NNC may lead to a larger DMT region in such scenarios. In strong and in very strong interference regimes, which are the focus of this study, the receiver jointly decodes both the interference and the desired message. In such situations, when CF is DMT optimal (see, e.g., Corollary \ref{cor:cf_optimallity}), clearly NNC cannot outperform CF. When CF is not DMT optimal, then NNC may indeed provide a better DMT performance.

    {\em The operational significance of DMT analysis in modern wireless communications systems:} An interesting aspect to investigate related to our DMT results for the ICR is their operational significance, as was done for point-to-point MIMO channels in \cite{LOZ:10}. The work \cite{LOZ:10} showed that in practical wideband operating scenarios with frequency diversity, link adaptation can be used to avoid outage in slowly fading channels, while in rapidly fading scenarios, hybrid automatic-repeat-request (HARQ) provides sufficient protection from outage. Thus, in both rapidly and slowly fading channels, the transmission scheme should utilize the available antennas for increasing the information rate (i.e., multiplexing gain) rather than for decreasing the probability of outage (i.e., diversity gain). Note however, that the analysis in \cite{LOZ:10} does not easily extend to the ICR studied in this paper for two reasons: First, we assume no transmitter CSI and no feedback, which precludes link adaptation as well as HARQ. In addition, our setup is a virtual MIMO channel, and hence, even with transmitter CSI everywhere, the ability to do link adaptation and HARQ over virtual links is not straightforward. Furthermore, the relay also complicates the analysis as it introduces multiple hops which were not present in \cite{LOZ:10}. We conjecture that as in \cite{LOZ:10}, if there is transmitter CSI and/or feedback in our model, then techniques such as link adaptation and HARQ will reduce or eliminate the need to use degrees of freedom for diversity in most settings, and hence typical operating scenarios will use most degrees of freedom for multiplexing. Making this conjecture rigorous, however, is a topic of future work.


\vspace{-0.2 cm}

\section{Summary}
\vspace{-0.2 cm}
\label{sec:summary}
In this work we studied the DMT performance of single-antenna Gaussian ICRs. We derived 
four achievable DMT regions based on CF, DF, and AF at the relay. Additionally, we derived conditions on the channel coefficients under which the optimal DMT is achieved with CF and with DF. In these scenarios, we showed that the optimal DMT of the ICR is the same as the optimal DMT for two parallel interference-free relay channels which means that a single relay can be DMT-optimal for both communicating pairs simultaneously, and that interference does not degrade the DMT performance when these conditions are satisfied. For the AF strategy we characterized scenarios in which the achievable DMT of the ICR is the same as the best DMT for two parallel, AF-based relay channels, and we showed that a single relay can assist both pairs to achieve this DMT simultaneously. These results demonstrate that adding a relay can substantially improve the DMT of interference channels, which gives a strong motivation for employing relay nodes in multi-user wireless networks that have to cope with interference.

\appendices
\newpage
\setcounter{equation}{0}
\numberwithin{equation}{section}
\section{Proof of the DMT outer bound of Proposition \ref{thm:theorem_cutset}}
\label{app:CUTSET}
We begin with the statement of a cut-set bound for the ICR, which is given by the following proposition:

\noindent
{\bf Proposition A:} {\em Let $\Rset_+$ denote the set of nonnegative real numbers. When $\uth$ is given and fixed, an outer bound on the achievable rate region is given by the following region:
\begin{subequations}
\label{eq:cut_set}
\begin{eqnarray}
    \mathcal{C}_{\mbox{\scriptsize outer-bound}}(\uth)\triangleq\bigcup_{f(x_1)f(x_2)f(x_3|x_1,x_2;\th_3,\th_{3,T})}\Big\{(R_1,R_2)&\in&\Rset_+^2:\nonumber\\
    \label{eq:cut_set_1} R_1     &\leq& I(X_1;Y_1,Y_3|X_2,X_3, \uth)\\
    \label{eq:cut_set_2} R_1     &\leq& I(X_1,X_3;Y_1|X_2, \uth)\\
    \label{eq:cut_set_3} R_2     &\leq& I(X_2;Y_2,Y_3|X_1,X_3, \uth)\\
    \label{eq:cut_set_4} R_2     &\leq& I(X_2,X_3;Y_2|X_1, \uth)\Big\}
\end{eqnarray}
\end{subequations}
}\vspace{-0.9 cm}
\begin{proof}
In order to establish the conditioning on the channel realization $\uth$ in \eqref{eq:cut_set}, we review the derivation of \cite[Theorem 15.10.1]{cover-thomas:it-book} starting from \cite[Eqn. (15.324)]{cover-thomas:it-book}. Enumerate the set of nodes in the network $\{$Tx$_1$, Tx$_2$, Relay, Rx$_1$, Rx$_2$$\}$ with $\{1,2,3,4,5\}\triangleq\mathfrak{N}$, respectively. Recall that $\mS$ and $\mS^c$ are a partition of the nodes in the network into two sets, and let $\mT$ denote the set of pairs of $(i,j)$ indexes s.t. $i\in\mS$ and $j\in\mS^c$. $\mT^c$ denotes the set of all the pairs of indexes in $\mathfrak{N}^2$ not in $\mT$. Let $W^{\mT}\triangleq\{W_{ij}\}_{(i,j)\in\mT}$, and let $X^{\mS}\triangleq\{X_{ik}\}_{i\in\mS}$. Define $\eps_n\triangleq\frac{1}{n}+\left(\sum_{i\in\mS, j\in\mS^c} R_{ij}\right)\P_e^{(n)}$ and note that $\eps_n\rightarrow 0$ as $n\rightarrow\infty$. Thus, for a set $\{R_{ij}\}_{i\in\mS, j\in\mS^c}$, we have
\begin{eqnarray*}
        n\sum_{i\in\mS, j\in\mS^c}\!\!\!\!R_{ij} & = & H\big(W^{(\mT)}|W^{(\mT^c)}\big)\\
        & \stackrel{(a)}{=} & H\big(W^{(\mT)}|W^{(\mT^c)},\uth\big)\\
        & = & H\big(W^{(\mT)}|W^{(\mT^c)},\uth\big) - H\big(W^{(\mT)}|W^{(\mT^c)},Y_1^{(\mS^c)}, Y_2^{(\mS^c)},...,Y_n^{(\mS^c)},\uth\big)\\
        & & \qquad + H\big(W^{(\mT)}|W^{(\mT^c)},Y_1^{(\mS^c)}, Y_2^{(\mS^c)},...,Y_n^{(\mS^c)},\uth\big)\\
        & = & I\big(W^{(\mT)};Y_1^{(\mS^c)}, Y_2^{(\mS^c)},...,Y_n^{(\mS^c)}|W^{(\mT^c)},\uth\big)\\
        & & \qquad + H\big(W^{(\mT)}|W^{(\mT^c)},Y_1^{(\mS^c)}, Y_2^{(\mS^c)},...,Y_n^{(\mS^c)},\uth\big)\\
        &\stackrel{(b)} {\le} & I\big(W^{(\mT)};Y_1^{(\mS^c)}, Y_2^{(\mS^c)},...,Y_n^{(\mS^c)}|W^{(\mT^c)},\uth\big) +n\eps_n\\
        & = & \sum_{i=1}^n I\big(W^{(\mT)};Y_i^{(\mS^c)}|Y_1^{(\mS^c)}, Y_2^{(\mS^c)},...,Y_{i-1}^{(\mS^c)},W^{(\mT^c)},\uth\big) +n\eps_n
 \end{eqnarray*}
 \begin{eqnarray*}
        \qquad\qquad\qquad\qquad\quad& = & \sum_{i=1}^n \Big[H\big(Y_i^{(\mS^c)}|Y_1^{(\mS^c)}, Y_2^{(\mS^c)},...,Y_{i-1}^{(\mS^c)},W^{(\mT^c)},\uth\big) \\
        & & \qquad - H\big(Y_i^{(\mS^c)}|Y_1^{(\mS^c)}, Y_2^{(\mS^c)},...,Y_{i-1}^{(\mS^c)},W^{(\mT^c)},W^{(\mT)},\uth\big) \Big]+n\eps_n\\
        & \stackrel{(c)}{\le} & \sum_{i=1}^n\Big[ H\big(Y_i^{(\mS^c)}|Y_1^{(\mS^c)}, Y_2^{(\mS^c)},...,Y_{i-1}^{(\mS^c)},W^{(\mT^c)},\uth\big) \\
        & & \qquad - H\big(Y_i^{(\mS^c)}|Y_1^{(\mS^c)}, Y_2^{(\mS^c)},...,Y_{i-1}^{(\mS^c)},W^{(\mT^c)},W^{(\mT)},X_i^{(\mS^c)} ,X_i^{(\mS)},\uth\big) \Big]+n\eps_n\\
        & \stackrel{(d)}{=} & \sum_{i=1}^n \Big[ H\big(Y_i^{(\mS^c)}|Y_1^{(\mS^c)}, Y_2^{(\mS^c)},...,Y_{i-1}^{(\mS^c)},W^{(\mT^c)}, X_i^{(\mS^c)} ,\uth\big) \\
        & & \qquad - H\big(Y_i^{(\mS^c)}|Y_1^{(\mS^c)}, Y_2^{(\mS^c)},...,Y_{i-1}^{(\mS^c)},W^{(\mT^c)},W^{(\mT)},X_i^{(\mS^c)} , X_i^{(\mS)},\uth\big)\Big] +n\eps_n\\
        & \stackrel{(e)}{=} & \sum_{i=1}^n \Big[H\big(Y_i^{(\mS^c)}|Y_1^{(\mS^c)}, Y_2^{(\mS^c)},...,Y_{i-1}^{(\mS^c)},W^{(\mT^c)}, X_i^{(\mS^c)} ,\uth\big) \\
        & & \qquad - H\big(Y_i^{(\mS^c)}|X_i^{(\mS^c)} , X_i^{(\mS)},\uth\big)\Big] +n\eps_n\\
        & \stackrel{(f)}{\le} & \sum_{i=1}^n \Big[H\big(Y_i^{(\mS^c)}| X_i^{(\mS^c)} ,\uth\big) - H\big(Y_i^{(\mS^c)}|X_i^{(\mS^c)} , X_i^{(\mS)},\uth\big)\Big] +n\eps_n\\
        & = & \sum_{i=1}^n I\big(X_i^{(\mS)}; Y_i^{(\mS^c)}| X_i^{(\mS^c)} ,\uth\big)  +n\eps_n
\end{eqnarray*}
where (a) follows as the messages are independent of the realization of the channel coefficients; (b) follows from the Fano's inequality; (c) follows as we added $X_i^{(\mS^c)} $ and $X_i^{(\mS)}$ to the conditioning in the second term in the summation and used the fact that conditioning reduces entropy; (d) follows as $X_i^{(\mS^c)}$ is uniquely determined
by the messages $W^{(\mT^c)}$, the channel outputs $Y_1^{(\mS^c)}$, $Y_2^{(\mS^c)},...$ ,$Y_{i-1}^{(\mS^c)}$, and the channel coefficients $\uth$, and therefore, adding $X_i^{(\mS)}$ to the conditioning of the first term of the summation does not change the entropy; (e) follows from the memorylessness of the channel; and (f) follows as conditioning reduces entropy.

Proceeding with the steps used to arrive from \cite[Eqn. (15.333)]{cover-thomas:it-book} to \cite[Eqn. (15.338)]{cover-thomas:it-book}, we obtain
\begin{equation}
\label{eq:eq28}
\sum_{i\in\mS, j\in\mS^c} R_{ij} \le I(X^{\mS};Y^{\mS^c}|X^{\mS^c},\uth),
\end{equation}
subject to some $P(X^{\mS\cup\mS^c}|\uth)$.

Thus, equations \eqref{eq:cut_set_1}-\eqref{eq:cut_set_4} are obtained by applying \eqref{eq:eq28} to the ICR for four partitions: $\mS=\{$Tx$_1\}$, $\mS=\{$Tx$_1$,Relay$\}$, $\mS=\{$Tx$_2\}$, and $\mS=\{$Tx$_2$,Relay$\}$, respectively. Note that as in the ICR model considered there is no feedback or CSI at the transmitters and the relay
has Rx-CSI and Tx-CSI on its incoming and outgoing links, the joint distribution for the cut-set bound can be decomposed into
\begin{equation*}
    p(x_1,x_2,x_3|\uth)=p(x_1)p(x_2)p(x_3|x_1,x_2;\th_3,\th_{3,T}),
\end{equation*}
where $(\th_3,\th_{3,T})$ should be taken as fixed throughout codeword transmission.
\end{proof}

Let $\mathcal{C}(\uth)$ denote the capacity region of the ICR for the channel coefficients $\uth$. Then $\mathcal{C}(\uth)\subseteq \mathcal{C}_{\mbox{\scriptsize outer-bound}}(\uth)$. Therefore, the outage probability corresponding to the outer bound, $\mathcal{C}_{\mbox{\scriptsize outer-bound}}(\uth)$, is a lower bound on the outage probability corresponding to $\mathcal{C}(\uth)$, i.e.,
\begin{equation*}
    \Pr\big((R_1,R_2)\notin \mathcal{C}_{\mbox{\scriptsize outer-bound}}(\utH)\big)\leq\Pr\big((R_1,R_2)\notin \mathcal{C}(\utH)\big),
\end{equation*}

It follows that the DMT region corresponding to the cut-set bound $\mathcal{C}_{\mbox{\scriptsize outer-bound}}(\uth)$ constitutes an outer bound on the achievable DMT region of the ICR. In the following, we characterize the DMT curves corresponding to \eqref{eq:cut_set_1}-\eqref{eq:cut_set_4}: Consider first the DMT corresponding to \eqref{eq:cut_set_1}, and let $R_{k,T}=r_k\log\SNR, k\in\{1,2\}$ denote the target rate for the pair Tx$_k$-Rx$_k$. The outage probability corresponding to \eqref{eq:cut_set_1} is defined as $\Pr(\outage_1^+)\triangleq\Pr\big(I(X_1;Y_1,Y_3|X_2,X_3,\uth)< r_1\log\SNR\big)$. Note that similar to \cite[Appendix A]{Zahavi:12}, $I(X_1;Y_1,Y_3|X_2,X_3,\uth)$ can be upper bounded as follows:
\begin{eqnarray*}
I(X_1;Y_1,Y_3|X_2,X_3,\uth)
&\le&\log\Big(1+\SNR|h_{11}|^2+\SNR^{\gamma}|h_{13}|^2\Big)\nonumber\\
\label{eq:eq14}&=&\log\Big(1+\SNR^{1-\theta_{11}}+\SNR^{\gamma-\theta_{13}}\Big)\\
&\triangleq&R^+_1(\theta_{11},\theta_{13}).
\end{eqnarray*}
For this upper bound we have,
    $\Pr(\outage_1^+)\ge\Pr(1+\SNR^{1-\theta_{11}}+\SNR^{\gamma-\theta_{13}}< \SNR^{r_1})\triangleq\Pr (\tilde{\outage}^+_1)$.
In order to calculate $\Pr (\tilde{\outage}^+_1)$, we follow similar steps as those used in \cite[Theorem 1]{Leveque:09}: Define $\theta_{kl}$ s.t. $|h_{kl}|^2 = \SNR^{-\theta_{kl}}$, where $k,l\in\{1,2,3\}, (k,l)\neq(3,3)$, and note that from \cite[Eqn. (5)]{Azarian:05} we obtain that since $h_{kl}$ are complex Normal RVs, then in the asymptotic case as $\SNR\rightarrow\infty$, the p.d.f. of $\theta_{kl}$ is equal to zero for all negative values of $\theta_{kl}$. Therefore, we consider only $\theta_{kl}\ge 0, k,l\in\{1,2,3\}, (k,l)\neq(3,3)$. Let
\begin{equation*}
\mathcal{D}_{r_1}\triangleq\Big\{\theta_{11}\geq 0, \theta_{13}\geq 0, R^+_1(\theta_{11},\theta_{13})<r_1\log\SNR \Big\}.
\end{equation*}
\noindent
Hence, using \cite[Eqn. (6)]{Azarian:05} we obtain that when $\SNR\rightarrow\infty$, then  $\Pr(\tilde{\outage}^+_1)=\Pr\Big((\theta_{11},\theta_{13})\in\mathcal{D}_{r_1}\Big)$ scales as
\begin{equation*}
    \Pr (\tilde{\outage}^+_1) \doteq \SNR^{-d_{\tilde{\outage}^+_1}(r_1)},
\end{equation*}
where
\begin{subequations}
\begin{eqnarray}
    d_{\tilde{\outage}^+_1}(r_1) = &&\min \;\;\;  \theta_{11}+\theta_{13}\nonumber\\
    \!\!\textrm{s.t. } && (1-\theta_{11})^+\le r_1\label{eq:eq291}\\
                   && (\gamma-\theta_{13})^+\le r_1\label{eq:eq292}\\
                   && \theta_{11}\ge0, \theta_{13}\ge 0.\label{eq:eq293}
\end{eqnarray}
\end{subequations}
As the constraints \eqref{eq:eq291}-\eqref{eq:eq293} can be rewritten as $\theta_{11} \ge (1-r_1)^+$ and $\theta_{13} \ge (\gamma-r_1)^+$, then the minimal sum equals to
\begin{equation*}
    d_{\tilde{\outage}^+_1}(r_1) = (1-r_1)^++(\gamma-r_1)^+ \triangleq d^+_{1}(r_1,r_2),
\end{equation*}
given in \eqref{eq:cut_dmt1}. Next, consider \eqref{eq:cut_set_2}. Define $\Pr(\outage_2^+)=\Pr\big(I(X_1,X_3;Y_1|X_2,\uth)< r_1\log\SNR\big)$. From \cite[Eqn. (A9)]{Zahavi:12}, we upper bound $I(X_1,X_3;Y_1|X_2,\uth)$ as follows:
\begin{eqnarray}
    I(X_1,X_3;Y_1|X_2,\uth)&\le& \log\big(1+\SNR|h_{11}|^2+\SNR^{\frac{1+\beta}{2}}h_{11}h_{31}^*+\SNR^{\frac{1+\beta}{2}}h_{31}h_{11}^*+\SNR^{\beta}|h_{31}|^2\big)\nonumber\\
    &=& \log\big(1+\SNR|h_{11}|^2+2\SNR^{\frac{1+\beta}{2}}\cdot\Real\{h_{11}h_{31}^*\}+\SNR^{\beta}|h_{31}|^2\big)\nonumber\\
    &\le& \log\big(1+\SNR|h_{11}|^2+2\SNR^{\frac{1+\beta}{2}}|h_{11}||h_{31}|+\SNR^{\beta}|h_{31}|^2\big)\nonumber\\
    \label{eq:eq13}&=&   \log\big(1+\SNR^{1-\theta_{11}}+2\SNR^{\frac{1-\theta_{11}+\beta-\theta_{31}}{2}}+\SNR^{\beta-\theta_{31}}\big).
\end{eqnarray}
Thus,
\begin{equation*}
    \Pr(\outage_2^+)\ge\Pr(1+\SNR^{1-\theta_{11}}+2\SNR^{\frac{1-\theta_{11}+ \beta-\theta_{31}}{2}}+\SNR^{\beta-\theta_{31}}< \SNR^{r_1})\triangleq \Pr (\tilde{\outage}^+_2),
\end{equation*}
Following \cite[Eqn. (6)]{Azarian:05}, we obtain that as $\SNR\rightarrow\infty$, then $\Pr (\tilde{\outage}^+_2)\doteq \SNR^{ -d_{\tilde{\outage}^+_2}(r_1)}$, where
\begin{eqnarray*}
    d_{\tilde{\outage}^+_2}(r_1) = && \min\;\;\;  \theta_{11}+\theta_{31}\\
    \textrm{s.t. } && (1-\theta_{11})^+\le r_1\\
                 && (\beta-\theta_{31})^+\le r_1\\
                 && \theta_{11}\ge0, \theta_{31}\ge0.
\end{eqnarray*}
Which follows since $\frac{(1-\theta_{11}+\beta-\theta_{31})^+}{2} \le \frac{(1-\theta_{11})^++(\beta-\theta_{31})^+}{2}$. For this minimization problem we obtain the solution $d_{\tilde{\outage}^+_2}(r_1) = (1-r_1)^+ + (\beta-r_1)^+ \triangleq d^+_{2}(r_1,r_2)$, given in \eqref{eq:cut_dmt2}. Following similar steps, we obtain the DMT bounds \eqref{eq:cut_dmt3} and \eqref{eq:cut_dmt4} from  \eqref{eq:cut_set_3} and \eqref{eq:cut_set_4}, respectively.

\section{Proof of Theorem \ref{thm:theorem_cf}}
\vspace{-0.2 cm}
\label{app:CFProof}
An achievable rate region for the ICR with only common messages and CF at the relay is given in \cite[Thm. 1]{Tian:11}. This region consists of
 all nonnegative rate pairs satisfying:
\begin{subequations}
\label{eq:cf_rate}
\begin{eqnarray}
    \label{eq:cf_rate11} R_1 &\leq& I(X_1;Y_1,\hat{Y}_3|X_2,X_3,\uth)\\
    \label{eq:cf_rate22} R_2 &\leq& I(X_2;Y_2,\hat{Y}_3|X_1,X_3,\uth)\\
    \label{eq:cf_rate3} R_1+R_2 &\leq& I(X_1,X_2;Y_1,\hat{Y}_3|X_3,\uth)\\
    \label{eq:cf_rate4} R_1+R_2 &\leq& I(X_1,X_2;Y_2,\hat{Y}_3|X_3,\uth),
\end{eqnarray}
\end{subequations}
subject to the constraints:
\vspace{-0.3 cm}
\begin{subequations}
\label{eq:cf_rate_cons}
\begin{eqnarray}
    \label{eq:cf_rate_cons1}I(X_3;Y_1|\uth)&\geq&I(Y_3;\hat{Y}_3|X_3,Y_1,\uth)\\
    \label{eq:cf_rate_cons2}I(X_3;Y_2|\uth)&\geq&I(Y_3;\hat{Y}_3|X_3,Y_2,\uth),
\end{eqnarray}
\end{subequations}
 for a joint distribution $f(x_1)f(x_2)f(x_3)f(y_1,y_2,y_3|x_1,x_2,x_3)f(\hat{y}_3|y_3,x_3)$, with complex Normal inputs $X_k\sim\CN(0,1)$,
 $k\in\{1,2,3\}$, and $\hat{Y}_3 = Y_3 + Z_{\textrm{Q}}$,  $Z_{\textrm{Q}} \sim \CN(0,N_{\textrm{Q}})$, independent of $\{Y_k\}_{k=1}^3$ and $\{X_k\}_{k=1}^3$
 where $N_{\textrm{Q}}$ is selected to satisfy \eqref{eq:cf_rate_cons}.
 Using the relationships $I(X_3;Y_k|\uth) = h(Y_k|\uth)\; - \; h(Y_k|X_3,\uth)$ and $I(Y_3;\hat{Y}_3|X_3,Y_k,\uth)=$
 $h(Y_k,\hat{Y}_3|X_3,\uth)$ $-\;h(Y_k|X_3,\uth)\;-\;\log\big((\pi e)N_{\textrm{Q}}\big), k\in\{1,2\}$, we can rewrite the constraints in \eqref{eq:cf_rate_cons1} and \eqref{eq:cf_rate_cons2} as:
\begin{equation}
    \label{eq:cf_rate_cons22}
    \log\big((\pi e)N_{\textrm{Q}}\big) \geq h(Y_k,\hat{Y}_3|X_3,\uth)-h(Y_k|\uth), \qquad k\in\{1,2\}.
\end{equation}

Next we find the smallest $N_{\textrm{Q}}$ that satisfies \eqref{eq:cf_rate_cons22}.
Starting with $k=1$, we write explicitly $h(Y_1|\uth)$ for mutually independent complex Normal channel inputs:
\begin{equation}
    \label{eq:eq9}
    h(Y_1|\uth)=\log\big((\pi e)(1+\SNR|h_{11}|^2+\SNR^{\alpha}|h_{21}|^2+\SNR^{\beta}|h_{31}|^2)\big).
\end{equation}
Defining
\begin{eqnarray*}
    \Hmat \triangleq \left[ \begin{array}{cc}
    \sqrt{\SNR}h_{11} & \; \; \;\; \sqrt{\SNR^\alpha}h_{21}\\
    \sqrt{\SNR^{\gamma}}h_{13} & \; \; \;\; \sqrt{\SNR^{\gamma}}h_{23}\end{array}\right], \;\; \Xvec \triangleq \left[ \begin{array}{c}
    X_1\\
    X_2\end{array} \right], \; \; \Zvec \triangleq \left[ \begin{array}{c}
    Z_1\\
    Z_3+Z_{\textrm{Q}}\end{array} \right],
\end{eqnarray*}
we obtain
\vspace{-0.4 cm}
\begin{eqnarray}
\!\!\! \!\!   h(Y_1,\hat{Y}_3|X_3,\uth) &=&h(\Hmat\cdot\Xvec+\Zvec)\nonumber\\
    &=&\log\!\!\Big(\!(\pi e)^2\big|\Hmat\cdot\cov(\Xvec)\cdot\Hmat^H+\cov(\Zvec)\big|\!\Big)\nonumber\\
    &\le& \log\!\!\Big(\!(\pi e)^2 \big((1+N_{\textrm{Q}})(1+\SNR|h_{11}|^2+\SNR^{\alpha}|h_{21}|^2)\nonumber\\
    &&\label{eq:eq10}\qquad\!\!+\SNR^{\gamma}|h_{13}|^2\!+\!\SNR^{\gamma}|h_{23}|^2\!+\!\SNR^{\gamma+\alpha}|h_{13}|^2|h_{21}|^2\!+\!\SNR^{\gamma+1}|h_{23}|^2|h_{11}|^2\big)\!\Big).
\end{eqnarray}
Combining \eqref{eq:eq9} and \eqref{eq:eq10} we conclude that \eqref{eq:cf_rate_cons22} is satisfied for $k=1$ if
\begin{equation*}
    N_{\textrm{Q}}\ge\frac{1+\SNR^{1-\theta_{11}}+\SNR^{\alpha-\theta_{21}}+\SNR^{\gamma-\theta_{13}}+ \SNR^{\gamma-\theta_{23}}+\SNR^{\gamma+\alpha-\theta_{13}-\theta_{21}}+\SNR^{\gamma+1-\theta_{11}-\theta_{23}}}{\SNR^{\beta-\theta_{31}}}.
\end{equation*}
where $\theta_{ij}$ is defined in Appendix \ref{app:CUTSET}. Note that since $\theta_{kl}\ge0, k,l\in\{1,2,3\}, (k,l)\neq(3,3)$, the above inequality is guaranteed if
\begin{equation*}
    N_{\textrm{Q}}\ge\frac{1+\SNR^{1}+\SNR^{\alpha}+\SNR^{\gamma}+ \SNR^{\gamma}+\SNR^{\gamma+\alpha}+\SNR^{\gamma+1}}{\SNR^{\beta-\theta_{31}}}.
\end{equation*}
Thus, we obtain
\begin{equation}
    \label{eq:eq25}
    N_{\textrm{Q}}\dot{=}\max\{\SNR^{\gamma+\alpha-\beta+\theta_{31}}, \SNR^{\gamma+1-\beta+\theta_{31}}\}.
\end{equation}
Using the same arguments for $k=2$ and combining with \eqref{eq:eq25}, we conclude that
\eqref{eq:cf_rate_cons22} is satisfied with
\begin{equation}
    \label{eq:eq25New}
    N_{\textrm{Q}}\dot{=}\max\{\SNR^{\gamma+\alpha-\beta+\theta_{31}}, \SNR^{\gamma+1-\beta+\theta_{31}}, \SNR^{\gamma+\alpha-\beta+\theta_{32}}, \SNR^{\gamma+1-\beta+\theta_{32}}\}.
\end{equation}
Note that since the expression for $N_{\textrm{Q}}$ in \eqref{eq:eq25New} depends on $\theta_{31}$ and $\theta_{32}$, then the relay must have Tx-CSI in order to compute $N_{\textrm{Q}}$. Using its Tx-CSI, the relay is able to identify the minimal compression required to be applied to its received signal, to permit reliable transmission of information on its received signal to the destinations. Also note that the degree of compression is proportional to the relative strength of the source-relay links compared to the strength of the relay-destination links (represented by $\frac{\SNR^{\gamma}}{\SNR^{\beta}}$). The Rx-CSI is needed at the relay to facilitate the use of Gaussian codebooks for compression.

Denote  the event that the $k$'th inequality in \eqref{eq:cf_rate11}-\eqref{eq:cf_rate4} is violated with $\outage_k^{\textrm{CF}}$. We first evaluate $\Pr(\outage_1^{\textrm{CF}})$ as follows:
\begin{eqnarray*}
    \Pr(\outage_1^{\textrm{CF}})
    &=& \Pr\big(h(Y_1,\hat{Y}_3|X_2,X_3,\uth)-h(Z_1,Z_3+Z_Q|\uth)<r_1\log\SNR\big)\\
    &=& \Pr\big(h(\sqrt{\SNR}h_{11}X_1\!+\!Z_1,\sqrt{\SNR^\gamma}h_{13}X_1\!+\!Z_3\!+\!Z_Q)\!-\!\log\big((\pi e)^2(1\!+\!N_{\textrm{Q}})\big)\!<\!r_1\log\SNR\big).
\end{eqnarray*}
Defining  ${\bf{H}} \triangleq \left[ \sqrt{\SNR}h_{11}, \sqrt{\SNR^\gamma}h_{13}\right]^T$, and $\Zvec \triangleq \left[   Z_1,   Z_3+Z_{\textrm{Q}} \right]^T$,
%
we can write
\vspace{-0.2 cm}
\begin{eqnarray*}
    \Pr(\outage_1^{\textrm{CF}}) 
    &=& \Pr\Big(h({\bf H}\cdot X_1+\Zvec)-h(\Zvec)<r_1\log\SNR\Big)\\
    &=& \Pr\left(\log\Big((\pi e)^2\big|{\bf H}{\bf H}^H+\cov(\Zvec)\big|\Big)-\log\Big((\pi e)^2(1+N_{\textrm{Q}})\Big)<r_1\log\SNR\right)\\
    &=& \Pr\big(1+\SNR|h_{11}|^2+\frac{\SNR^{\gamma}|h_{13}|^2}{1+N_{\textrm{Q}}}<\SNR^{r_1}\big).
\end{eqnarray*}
Next, we write
\vspace{-0.3 cm}
\begin{eqnarray*}
    1+N_{\textrm{Q}}&\dot{=}&\max\{\SNR^{\gamma+\alpha-\beta+\theta_{31}}, \SNR^{\gamma+1-\beta+\theta_{31}}, \SNR^{\gamma+\alpha-\beta+\theta_{32}}, \SNR^{\gamma+1-\beta+\theta_{32}},\SNR^0\}\\
    &=&\max\{\SNR^{(\gamma+\alpha-\beta+\theta_{31})^+}, \SNR^{(\gamma+1-\beta+\theta_{31})^+}, \SNR^{(\gamma+\alpha-\beta+\theta_{32})^+}, \SNR^{(\gamma+1-\beta+\theta_{32})^+}\}.
\end{eqnarray*}
Hence, as in \cite[Theorem 1]{Leveque:09}, by following similar steps as those used in Appendix \ref{app:CUTSET}, the DMT corresponding to the event $\outage_1^{\textrm{CF}}$ can be calculated by solving the following minimization problem:
\begin{subequations}
\label{eq:eq30}
\begin{eqnarray}
    \min &&\; \theta_{11}+\theta_{13}+\theta_{31}+\theta_{32}\\
    \textrm{s.t.} &&\Big(1-\theta_{11}\Big)^+ \leq r_1,\\
                  &&\Big(\gamma-\theta_{13}-\max\big\{(\gamma+\alpha-\beta+\theta_{31})^+,(\gamma+1-\beta+\theta_{31})^+,\nonumber\\
                  &&\qquad\qquad\qquad\qquad (\gamma+\alpha-\beta+\theta_{32})^+,(\gamma+1-\beta+\theta_{32})^+\big\}\Big)^+ \leq r_1,\\
                  && \theta_{11}\ge0,\theta_{13}\ge0,\theta_{31}\ge0,\theta_{32}\ge0.
\end{eqnarray}
\end{subequations}
First, consider the case where $\alpha>1$. The case for $\alpha\le 1$ can be solved using similar arguments. For simplicity, define $\phi(\theta)\triangleq(\gamma+\alpha-\beta+\theta)^+$. Next, note that given $\theta_{31}$ and $\theta_{32}$, the optimal values for $\theta_{11}$ and $\theta_{13}$ can be obtained as
\begin{eqnarray*}
    \theta_{11}&=&(1-r_1)^+\\
    \theta_{13}&=&\Big(\gamma-\max\big\{\phi(\theta_{31}),\phi(\theta_{32})\big\}-r_1\Big)^+
\end{eqnarray*}
Define
\begin{equation*}
    \hat{f}(\theta_{31},\theta_{32})\triangleq(1-r_1)^+ + \Big(\gamma-\max\big\{\phi(\theta_{31}),\phi(\theta_{32})\big\}-r_1\Big)^+ +\theta_{31}+\theta_{32}
\end{equation*}
Thus, the optimization problem in \eqref{eq:eq30} can be rewritten as
\begin{subequations}
\label{eq:OptCFInd22}
\begin{eqnarray}
    \substack{\min\\ \theta_{31},\theta_{32}} &&\; \hat{f}(\theta_{31},\theta_{32})\\
    \textrm{s.t. }&&\     \theta_{31}\ge 0, \theta_{32} \ge 0.
\end{eqnarray}
\end{subequations}
Searching over all possible values of $\phi(\theta_{31})$ and $\phi(\theta_{32})$, there are four possible cases:
\begin{enumerate}
    \item $\gamma+\alpha-\beta+\theta_{31}\le0$ and $\gamma+\alpha-\beta+\theta_{32}\le0$: In this case we obtain $\phi(\theta_{31})=\phi(\theta_{32})=0$, for which we have
    \begin{equation*}
        \hat{f}(\theta_{31},\theta_{32})=(1-r_1)^++(\gamma-r_1)^++\theta_{31}+\theta_{32}.
    \end{equation*}
        It follows that in this case $\hat{f}(\theta_{31},\theta_{32})$ is a monotonically increasing function of $\theta_{31}$ and $\theta_{32}$.

    \item $\gamma+\alpha-\beta+\theta_{31}>0$ and $\gamma+\alpha-\beta+\theta_{32}\le0$: In this case we obtain $\phi(\theta_{31})=\gamma+\alpha-\beta+\theta_{31}$ and $\phi(\theta_{32})=0$. Thus,
    \begin{equation*}
        \hat{f}(\theta_{31},\theta_{32})=(1-r_1)^++\big(\gamma-\phi(\theta_{31})-r_1\big)^++\theta_{31}+\theta_{32}.
    \end{equation*}
    Here, there are two possibilities:
    \begin{itemize}
        \item $\gamma-\phi(\theta_{31})-r_1\le0$ for which we obtain
        \begin{equation*}
            \hat{f}(\theta_{31},\theta_{32})=(1-r_1)^++\theta_{31}+\theta_{32}.
        \end{equation*}
         It follows that in this case $\hat{f}(\theta_{31},\theta_{32})$ is again a monotonically increasing function of $\theta_{31}$ and $\theta_{32}$.
         \item $\gamma-\phi(\theta_{31})-r_1>0$ for which we obtain
        \begin{eqnarray*}
            \hat{f}(\theta_{31},\theta_{32})&=&(1-r_1)^++\big(\gamma-\phi(\theta_{31})-r_1\big)+\theta_{31}+\theta_{32}\\            &=&(1-r_1)^++\big(\gamma-(\gamma+\alpha-\beta+\theta_{31})-r_1\big)+\theta_{31}+\theta_{32}\\
            &=&(1-r_1)^++\big(\gamma-(\gamma+\alpha-\beta)-r_1\big)+\theta_{32}.
        \end{eqnarray*}
        It follows that in this case $\hat{f}(\theta_{31},\theta_{32})$ does not depend on $\theta_{31}$ but it is a  monotonically increasing function of $\theta_{32}$.
    \end{itemize}
    \item $\gamma+\alpha-\beta+\theta_{31}\le0$ and $\gamma+\alpha-\beta+\theta_{32}>0$: Following steps similar to the previous case, we conclude that $\hat{f}(\theta_{31},\theta_{32})$ is either a  monotonically increasing function of $\theta_{31}$ and $\theta_{32}$, or it does not depend on $\theta_{32}$ and is a  monotonically increasing function of $\theta_{31}$.
    \item $\gamma+\alpha-\beta+\theta_{31}>0$ and $\gamma+\alpha-\beta+\theta_{32}>0$: In this case we obtain $\phi(\theta_{31})=\gamma+\alpha-\beta+\theta_{31}$ and $\phi(\theta_{32})=\gamma+\alpha-\beta+\theta_{32}$. In this scenario,
    \begin{equation*}
        \hat{f}(\theta_{31},\theta_{32})=(1-r_1)^++\Big(\gamma-\max\big\{\phi(\theta_{31}),\phi(\theta_{32})\big\}-r_1\Big)^++\theta_{31}+\theta_{32},
    \end{equation*}
    Depending on whether $\phi(\theta_{31})>\phi(\theta_{32})$ or $\phi(\theta_{31})\le\phi(\theta_{32})$, this case becomes the same as the second or the third case, respectively.
\end{enumerate}
We conclude that for the optimization problem in \eqref{eq:OptCFInd22},  $\hat{f}(\theta_{31},\theta_{32})$ is either a  monotonically increasing function of both $\theta_{31}$ and $\theta_{32}$, or it does not depend on one and is a monotonically increasing function of the other. Thus, the optimal $\theta_{31}$ and $\theta_{32}$ for this optimization problem are zero. Note that for $\theta_{31}=\theta_{32}=0$, we have that $\phi(\theta_{31})=\phi(\theta_{32})=(\gamma+\alpha-\beta)^+$, hence, the optimal solution to \eqref{eq:OptCFInd22} is
\begin{equation*}
    \substack{\min\\ \theta_{31},\theta_{32}} \hat{f}(\theta_{31},\theta_{32}) = (1-r_1)^++(\gamma-(\gamma+\alpha-\beta)^+-r_1)^+.
\end{equation*}
Repeating the same steps for the case $\alpha\le 1$, we obtain
\begin{equation}
\label{eq:eq11}
    d_{1,\textrm{CF}}^{-} =
    \begin{cases}
        (1-r_1)^+ + \left(\gamma-(\gamma+\alpha-\beta)^+ -r_1\right)^+ & \qquad\qquad\qquad\qquad\qquad ,\alpha>1\\
        (1-r_1)^+ + \left(\gamma-(\gamma+1-\beta)^+ -r_1\right)^+        & \qquad\qquad\qquad\qquad\qquad ,\alpha\le1,
    \end{cases}
\end{equation}
which is \eqref{eq:cf_dmt1}.
Similarly, we obtain the achievable DMT \eqref{eq:cf_dmt2} by calculating the probability of the outage event $\outage_2^{\textrm{CF}}$ which follows from the rate constraint \eqref{eq:cf_rate22}.
Consider next the outage probability $\Pr(\outage_3^{\textrm{CF}})$ corresponding to the rate constraint \eqref{eq:cf_rate3}. Defining
\begin{eqnarray*}
    \Hmat \triangleq \left[ \begin{array}{cc}
    \sqrt{\SNR}h_{11} & \;\;\;\;\sqrt{\SNR^\alpha}h_{21}\\
    \sqrt{\SNR^\gamma}h_{13} & \;\;\;\;\sqrt{\SNR^\gamma}h_{23}\end{array}\right],\;\;\Xvec\triangleq \left[ \begin{array}{c}
    X_1\\
    X_2\end{array} \right], \;\;\Zvec \triangleq \left[ \begin{array}{c}
    Z_1\\
    Z_3+Z_{\textrm{Q}}\end{array} \right],
\end{eqnarray*}
we write
\begin{eqnarray*}
    \Pr(\outage_3^{\textrm{CF}})
    &=& \Pr\left(I(X_1,X_2;Y_1,\hat{Y}_3|X_3.\uth)<(r_1+r_2)\log\SNR\right)\\
    &=& \Pr\Big(h(\Hmat\cdot \Xvec+\Zvec|\uth)-\log\big((\pi e)^2(1+N_{\textrm{Q}})\big)<(r_1+r_2)\log\SNR\Big)\\
    &=& \Pr\left(\log\Big((\pi e)^2\big|\Hmat\Hmat^H+\cov(\Zvec)\big|\Big)-\log\Big((\pi e)^2(1+N_{\textrm{Q}})\Big)<(r_1+r_2)\log\SNR\right)\\
    &=& \Pr\Bigg(1+\SNR|h_{11}|^2+\SNR^\alpha |h_{21}|^2+\frac{\SNR^{\gamma}|h_{13}|^2+\SNR^{\gamma}|h_{23}|^2}{1+N_{\textrm{Q}}}+\\
    &&\qquad\qquad\qquad\qquad\qquad\qquad \frac{\left|\SNR^{\frac{\gamma+\alpha}{2}}h_{13}h_{21}^*-\SNR^{\frac{\gamma+1}{2}} h_{23}h_{11}^*\right|^2}{1+N_{\textrm{Q}}}<\SNR^{r_1+r_2}\Bigg)\\
    &\leq&\Pr\Bigg(1+\SNR|h_{11}|^2+\SNR^\alpha|h_{21}|^2+\frac{\SNR^{\gamma}|h_{13}|^2}{1+N_{\textrm{Q}}}<\SNR^{r_1+r_2}\Bigg).
\end{eqnarray*}
Thus, as in \cite[Theorem 1]{Leveque:09}, by following similar steps as those used in Appendix \ref{app:CUTSET}, a lower bound on the DMT is obtained by considering the following minimization problem:
\begin{eqnarray*}
    \min & &\;\; \theta_{11}+\theta_{21}+\theta_{13}+\theta_{31}+\theta_{32}\\
    \textrm{s.t.}
    && \big(1-\theta_{11}\big)^+ \leq r_1+r_2, \qquad      \big(\alpha-\theta_{21}\big)^+ \leq r_1+r_2,\\
    && \Big(\gamma-\theta_{13}-\max\big\{(\gamma+\alpha-\beta+\theta_{31})^+,(\gamma+1-\beta+\theta_{31})^+,\\
    &&\qquad\qquad\qquad\qquad(\gamma+\alpha-\beta+\theta_{32})^+,(\gamma+1-\beta+\theta_{32})^+\big\}\Big)^+ \leq r_1+r_2,\\
    && \theta_{k,l}\ge0, \qquad k,l\in\{1,2,3\}, (k,l)\neq(3,3).
\end{eqnarray*}
Similar to the previous case, the minimal solution is obtained at $\theta_{31}=\theta_{32}=0$. The resulting DMT relationship is characterized by $d_{3,\textrm{CF}}^{-}(r_1,r_2)$ stated in \eqref{eq:cf_dmt3}. An identical DMT expression is obtained from the analysis of $\Pr(\outage_4^{\textrm{CF}})$. This completes the proof.
\tend

\vspace{-0.2 cm}

\section{Proof of Theorem \ref{thm:theorem_df}}
\label{app:DFProof}
The achievability scheme is based on employing DF at the relay and using i.i.d. codebooks generated according to mutually independent,
zero-mean, complex Normal channel inputs. Let $\outage_R$ denote the outage event at the relay, i.e., the event that the relay fails to decode, and let $\outage_R^c$ denote its complement.
Then, the probability of an outage for the ICR can be evaluated as follows:
\begin{equation}
    \label{eq:DMT_DF}
    \Pr(\mbox{outage}) = \Pr(\mbox{outage}|\outage_R)\Pr(\outage_R) + \Pr(\mbox{outage}|\outage_R^c)\Pr(\outage_R^c),
\end{equation}
Similarly to \cite[Eqn. (A1)]{Dabora:12}, an achievable rate region for decoding at the relay is given by all nonnegative pairs $(R_1,R_2)$ satisfying
\vspace{-1.0cm}

\begin{eqnarray*}
    R_1 &\leq& I(X_1;Y_3|X_2,X_3,\th_3,\th_{3,T})=\log(1+\SNR^{\gamma}|h_{13}|^2)\\
    R_2 &\leq& I(X_2;Y_3|X_1,X_3,\th_3,\th_{3,T})=\log(1+\SNR^{\gamma}|h_{23}|^2)\\
    R_1+R_2 &\leq& I(X_1,X_2;Y_3|X_3,\th_3,\th_{3,T})=\log(1+\SNR^{\gamma}|h_{13}|^2+\SNR^{\gamma}|h_{23}|^2).
\end{eqnarray*}
\vspace{-1.0cm}

\noindent
The probability of an outage at the relay, $\Pr(\outage_R)$, corresponds to the event that at least one of the above inequalities is not satisfied. Applying similar techniques as in Appendix \ref{app:CUTSET}, we obtain
\begin{eqnarray*}
    \label{eq:outage_at_relay1}
    \Pr(\outage_R) &\dot{\leq}& \SNR^{-\min\left\{(\gamma-r_1)^+,(\gamma-r_2)^+,2(\gamma-r_1-r_2)^+\right\}} \triangleq \SNR^{-d_{\mbox{\tiny Relay}}(r_1,r_2)}.
\end{eqnarray*}
\vspace{-1cm}

\noindent
Thus, similarly to \cite[Appendix II]{Yuksel:07} it can be shown that at asymptotically high SNR
\begin{subequations}
\label{eq:outage_at_relay2}
\begin{eqnarray}
    &&\Pr(\outage_R)\doteq 
    \begin{cases}
        \SNR^{-\min\left\{(\gamma-r_1)^+,(\gamma-r_2)^+,2(\gamma-r_1-r_2)^+\right\}} & \quad r_1+r_2<\gamma\\
        1   & \quad r_1+r_2\ge \gamma.
    \end{cases}
\end{eqnarray}
\end{subequations}
When the relay fails to decode it remains silent at the next transmission block, and hence, the ICR specializes to the IC in such situations (recall that the receivers have Rx-CSI). Note that the destinations can be made aware of this (see, e.g., \cite[Appendix II]{Yuksel:07}) via a single bit sent from the relay at no rate cost asymptotically. It follows that when the relay fails to decodes, each receiver jointly decodes both messages based on its received signal which is a sum of the desired signal and the interfering signal.
As the sources use mutually independent, i.i.d. generated codebooks, then using the error analysis of \cite[Section IV-D]{Kramer:06} without superposition encoding
(i.e., setting $T_1=T_2=0$ in \cite[Section IV-D]{Kramer:06}) we obtain the following rate region:

\vspace{-1cm}
\begin{subequations}
\label{eq:eq31}
\begin{eqnarray}
    R_1 &\leq& I(X_1;Y_1|X_2,\th_1)=\log\big(1+\SNR|h_{11}|^2\big)\\
    R_2 &\leq& I(X_2;Y_2|X_1,\th_2)=\log\big(1+\SNR|h_{22}|^2\big)\\
    R_1+R_2 &\leq& I(X_1,X_2;Y_1|\th_1)=\log\big(1+\SNR|h_{11}|^2+\SNR^{\alpha}|h_{21}|^2\big)\\
    R_1+R_2 &\leq& I(X_1,X_2;Y_2|\th_2)=\log\big(1+\SNR^{\alpha}|h_{12}|^2+\SNR|h_{22}|^2\big).
\end{eqnarray}
\end{subequations}
\vspace{-0.3cm}
\noindent

Next, denote the target rates $R_{1,T}=r_1\log\SNR$ and $R_{2,T}=r_2\log\SNR$. An outage occurs if at least one of the inequalities in \eqref{eq:eq31} is not satisfied. Denote the probability of outage at the destinations given that the relay fails to decode with $\Pr(\mbox{outage}|\outage_R)$. Since this corresponds to an outage event for the IC (\eqref{eq:eq31}), then we write

\begin{equation}
    \label{eq:DMT_relay_doesnt_decodes}
    \Pr(\mbox{outage}|\outage_R) \doteq \SNR^{-d_{\mbox{\tiny IC}}(r_1,r_2)},
\end{equation}

\vspace{-0.2cm}

\noindent
where $d_{\mbox{\tiny IC}}(r_1,r_2)$ is the achievable DMT of the IC (without relay) corresponding to \eqref{eq:eq31}. From \cite[Theorem 1]{Bolcskei:09} it follows that $d_{\mbox{\tiny IC}}(r_1,r_2)$ is given by

\vspace{-0.6 cm}
\begin{equation*}
    d_{\mbox{\tiny IC}}(r_1,r_2)=\min\left\{(1-r_1)^+,(1-r_2)^+,(1-r_1-r_2)^+ + (\alpha-r_1-r_2)^+\right\}.
\end{equation*}

\vspace{-0.2 cm}

\noindent
Next, we consider the case where the relay decodes {\em both messages} successfully.\! Using relay encoding and receiver decoding as in \cite[\!Appendix A]{Dabora:12},\!
and neglecting errors in decoding the interfering message as in \cite[\!Section IV-D]{Kramer:06},\!
an achievable rate region of the ICR is obtained as:
\begin{eqnarray*}
    R_1 &\leq & I(X_1,X_3;Y_1|X_2,\th_1)     = \log(1+\SNR|h_{11}|^2+\SNR^\beta |h_{31}|^2)\\
    R_2 &\leq & I(X_2,X_3;Y_2|X_1,\th_2)     = \log(1+\SNR|h_{22}|^2+\SNR^\beta |h_{32}|^2)\\
    R_1+R_2 &\leq & I(X_1,X_2,X_3;Y_1|\th_1) = \log(1+\SNR|h_{11}|^2+\SNR^\alpha |h_{21}|^2+\SNR^\beta |h_{31}|^2)\\
    R_1+R_2 &\leq & I(X_1,X_2,X_3;Y_2|\th_2) = \log(1+\SNR|h_{22}|^2+\SNR^\alpha |h_{12}|^2+\SNR^\beta |h_{32}|^2).
\end{eqnarray*}

\vspace{-0.1 cm}

\noindent
Evaluating the DMT region corresponding to the above rate region similar to Appendix \ref{app:CUTSET}, it follows that
when the relay decodes both messages successfully, the probability of outage is given by

\vspace{-0.5cm}
\begin{equation}
    \label{eq:DMT_relay_decodes}
    \Pr(\mbox{outage}|\outage_R^c) \doteq \SNR^{-d_{\mbox{\tiny Coop.}}(r_1,r_2)},
\end{equation}
\vspace{-1.2 cm}

\noindent
where
\vspace{-0.3 cm}
\begin{eqnarray*}
    d_{\mbox{\tiny Coop.}}(r_1,r_2) & = & \min\big\{(1-r_1)^++(\beta-r_1)^+,(1-r_2)^++(\beta-r_2)^+,\\
    &&\qquad\qquad
    (1-r_1-r_2)^++(\alpha-r_1-r_2)^++(\beta-r_1-r_2)^+\big\}.
\end{eqnarray*}

\vspace{-0.2 cm}

\noindent
Finally, by substituting \eqref{eq:outage_at_relay2}, \eqref{eq:DMT_relay_doesnt_decodes}, and \eqref{eq:DMT_relay_decodes} into \eqref{eq:DMT_DF}, we obtain
\vspace{-0.9 cm}

\begin{eqnarray*}
\Pr(\mbox{outage})
&=& \Pr(\mbox{outage}|\outage_R)\Pr(\outage_R) + \Pr(\mbox{outage}|\outage_R^c)\Pr(\outage_R^c)\\
&\dot{\leq}& \begin{cases}
          \SNR^{-d_{\mbox{\tiny IC}}(r_1,r_2)}\SNR^{-d_{\mbox{\tiny Relay}}(r_1,r_2)}+\SNR^{-d_{\mbox{\tiny Coop.}}(r_1,r_2)} & \qquad \quad ,r_1+r_2<\gamma\\
          \SNR^{-d_{\mbox{\tiny IC}}(r_1,r_2)} & \quad \qquad ,r_1+r_2\ge \gamma
         \end{cases}
\end{eqnarray*}
which corresponds to the DMT region of Theorem \ref{thm:theorem_df}.

\tend

\section{Proof of Theorem \ref{thm:theorem_af}}
\label{AFFDProof}
We first derive an achievable rate region for the ICR with a full-duplex relay employing the AF scheme, and then we evaluate the DMT region obtained with this scheme.

\subsection{An Achievable Rate Region}
\vspace{-0.2cm}
\label{sec:AFregionHD}
Transmission is carried out in groups of $B-1$ messages. Each message is transmitted via a codeword of length $n$ channel symbols and an entire group of $B-1$ messages is transmitted using $nB$ channel symbols. Let $R_{k}$ denote the rate for the pair Tx$_k$-Rx$_k$. Then, the overall rate is $\frac{B-1}{B}R_{k}$ which approaches $R_{k}$ as $B$ increases. Let $\M_k \triangleq \{1,2,...,2^{n R_k}\}$ denote the message set for Tx$_k, k\in\{1,2\}$, and let the codebook for user $k$ be the set $\big\{\xvec_k(m_k)\big\}_{m_k\in\M_k}$
of mutually independent codewords selected according to $f\big(\xvec_k(m_k)\big)=\prod_{i=1}^nf_{X_k}(x_{k,i}(m_k))$.
At block $b$, Tx$_k$ sends a new message $m_{k,b}\in\M_k$ by transmitting
$\xvec_k(m_{k,b})\triangleq\xvec^{(b)}_k, k\in\{1,2\}$, and the relay transmits a scaled version of the signal received at the previous block, i.e., at block $b-1$. Let $H_{kl}^{(b)}$ denote the channel coefficient
$H_{kl}$ at block $b$ and let $G_R^{(b)}$ denote
the scaling applied by the relay at block $b$. $G_R^{(b)}$ is determined solely based on the Rx-CSI at the relay. The relationship between the channel inputs and outputs at the $i$'th
symbol of block $b$, $i\in\{1,2,...,n\}$, is given by:
\begin{subequations}
\begin{eqnarray}
\!\!\!\!    Y_{1,i}^{(b)} & = & \sqrt{\SNR}H_{11}^{(b)} X_{1,i}^{(b)} + \sqrt{\SNR^{\alpha}}H_{21}^{(b)} X_{2,i}^{(b)}\nonumber\\
    &&\quad+ \sqrt{\SNR^{\beta}}H_{31}^{(b)} G_R^{(b)} \!\left(\sqrt{\SNR}H_{13}^{(b-1)} X_{1,i}^{(b-1)}\! + \sqrt{\SNR}H_{23}^{(b-1)} X_{2,i}^{(b-1)}\! + Z_{3,i}^{(b-1)}\right)\! +Z_{1,i}^{(b)}\\
\!\!\!\!    Y_{2,i}^{(b)} & = & \sqrt{\SNR^{\alpha}} H_{12}^{(b)} X_{1,i}^{(b)} + \sqrt{\SNR}H_{22}^{(b)} X_{2,i}^{(b)}\nonumber\\
    &&\quad+ \sqrt{\SNR^{\beta}}H_{32}^{(b)} G_R^{(b)} \!\left(\sqrt{\SNR}H_{13}^{(b-1)} X_{1,i}^{(b-1)}\! + \sqrt{\SNR}H_{23}^{(b-1)} X_{2,i}^{(b-1)}\! + Z_{3,i}^{(b-1)}\right)\! +Z_{2,i}^{(b)}\;\;\;\;\;\\
\!\!\!\!    Y_{3,i}^{(b)} & = & \sqrt{\SNR}H_{13}^{(b)} X_{1,i}^{(b)} + \sqrt{\SNR}H_{23}^{(b)} X_{2,i}^{(b)} + Z_{3,i}^{(b)}.
\end{eqnarray}
\end{subequations}
Let $\tilde{H}_{k}^{(b)}\triangleq(H_{1k}^{(b)},H_{2k}^{(b)},H_{3k}^{(b)})\in\Cset^3$ denote the available Rx-CSI at Rx$_k$ at block $b, k\in\{1,2\}$, and let $\tilde{H}_{3}^{(b)}\triangleq(H_{13}^{(b)},H_{23}^{(b)})\in\Cset^2$ denote the Rx-CSI at the relay at block $b$. As the receivers know the Rx-CSI at the relay, they know $G_R^{(b)}$ as well.

The transmission scheme is inspired by the D-BLAST scheme \cite[Ch. 10.6.4]{Goldsmith:Book}: Rx$_k$ decodes $m_{k,b}$ at block $b+1$ as follows: Rx$_k$ {\em first decodes the interference in blocks $b$ and $b+1$} by treating the entire signal from the relay and its own desired signal as i.i.d. additive white Gaussian noise. This can be done reliably if $n$ is large enough and

 \vspace{-1 cm}

\begin{subequations}
\label{eq:AFVSIICancel}
\begin{eqnarray}
    R_{1}^{(b)} &\le& I\big(X^{(b)}_{1};Y_2^{(b)}|\th_{2}^{(b)}\big)\\
    R_{2}^{(b)} &\le& I\big(X^{(b)}_{2};Y_1^{(b)}|\th_{1}^{(b)}\big).\vspace{-0.3 cm}
\end{eqnarray}
\end{subequations}

Let $(\hat{\hat{m}}_{2,b},\hat{\hat{m}}_{2,b+1})$ denote the estimation of $(m_{2,b},m_{2,b+1})$ at Rx$_1$.
Rx$_1$ now jointly processes $\big(\yvec_{1}^{(b)}, \yvec_{1}^{(b+1)}\big)$ to decode $m_{1,b}$ as follows:  From decoding at the previous block, Rx$_1$ has an estimation of $(m_{1,b-1},m_{2,b-1})$ denoted by $(\hat{m}_{1,b-1},\hat{m}_{2,b-1})$. Rx$_1$ now generates the signal
\begin{subequations}
\label{eq:eq32}
\begin{eqnarray}
    \tilde{y}_{1,i}^{(b)}&=&y_{1,i}^{(b)}\!-\!\sqrt{\SNR^{1+\beta}} G_R^{(b)} h_{31}^{(b)} \Big(h_{13}^{(b-1)}x_{1,i}(\hat{m}_{1,b-1})+h_{23}^{(b-1)}x_{2,i}(\hat{m}_{2,b-1})\Big)\!-\!\sqrt{\SNR^{\alpha}}h_{21}^{(b)}x_{2,i}(\hat{\hat{m}}_{2,b})\\
    \tilde{y}_{1,i}^{(b+1)}&=&y_{1,i}^{(b+1)}-\sqrt{\SNR^{1+\beta}} G_R^{(b+1)} h_{31}^{(b+1)} h_{23}^{(b)}x_{2,i}(\hat{\hat{m}}_{2,b}) -\sqrt{\SNR^{\alpha}} h_{21}^{(b+1)}x_{2,i}(\hat{\hat{m}}_{2,b+1}),
\end{eqnarray}
\end{subequations}
$i\in\{1,2,3,...,n\}$. Assuming correct decoding of $(m_{2,b}, m_{2,b+1},m_{1,b-1},m_{2,b-1})$ at Rx$_1$, we can write $\tilde{y}_{1,i}^{(b)}$ and $\tilde{y}_{1,i}^{(b+1)}$ as:
\begin{eqnarray*}
    \tilde{y}_{1,i}^{(b)}&=&  \sqrt{\SNR}h_{11}^{(b)} x_{1,i}(m_{1,b})+\sqrt{\SNR^{\beta}} G_R^{(b)} h_{31}^{(b)} z^{(b-1)}_{3,i}+z^{(b)}_{1,i}\\
    \tilde{y}_{1,i}^{(b+1)}
    &=&\sqrt{\SNR}h_{11}^{(b+1)} x_{1,i}(m_{1,b+1})+\sqrt{\SNR^{\beta}} G_R^{(b+1)} h_{31}^{(b+1)} \Big(\sqrt{\SNR}h_{13}^{(b)}x_{1,i}(m_{1,b})+z^{(b)}_{3,i}\Big)+z^{(b+1)}_{1,i}.
\end{eqnarray*}
It follows that $\tilde{\yvec}_{1}^{(b)}$ is an interference free, noisy version of the desired signal at Rx$_1$ at block $b$ ($m_{1,b}$), and $\tilde{\yvec}_{1}^{(b+1)}$ is a noisy version of the codeword corresponding to message $m_{1,b}$ which includes interference caused by transmission of $m_{1,b+1}$ at block $b+1$. Note that this interference cannot be cancelled since Rx$_1$ decodes $m_{1,b+1}$ at block $b+2$. We conclude that $m_{1,b}$ can be reliably decoded if $n$ is large enough and
\begin{subequations}
\label{eq:AFVSIRegion}
\begin{equation}
    R_{1}^{(b)} \le I\big(X_{1}^{(b)};\tilde{Y}_{1}^{(b)},\tilde{Y}_{1}^{(b+1)}|\th_{1}^{(b)},\th_{1}^{(b+1)},\th_3^{(b)}\big).
\end{equation}
Following similar steps, we obtain that Rx$_2$ can decode $m_{2,b}$ reliably if
\begin{equation}
    R_{2}^{(b)} \le I\big(X_{2}^{(b)};\tilde{Y}_{2}^{(b)},\tilde{Y}_{2}^{(b+1)}|\th_{2}^{(b)},\th_{2}^{(b+1)},\th_3^{(b)}\big),
\end{equation}
\end{subequations}
where $\big(\tilde{Y}_{2}^{(b)},\tilde{Y}_{2}^{(b+1)}\big)$ is defined similarly to $\big(\tilde{Y}_{1}^{(b)},\tilde{Y}_{1}^{(b+1)}\big)$. Note that in general, in order to maximize the achievable rate region, we should use the values of $G^{(b)}_R$ and $G^{(b+1)}_R$ which maximize \eqref{eq:AFVSIICancel} and \eqref{eq:AFVSIRegion}. However, since this computation is very involved, we take here a suboptimal approach: Since the power of the relay is limited to $1$, then $\left(G_R^{(b)}\right)^2\le\frac{1}{1+{\sSNR}|h_{13}^{(b-1)}|^2+{\sSNR}|h_{23}^{(b-1)}|^2}$. Hence, as we define $|h_{kl}^{(b-1)}|^2=\SNR^{-\theta_{kl}^{(b-1)}}$ for $\theta_{kl}^{(b-1)}\ge0$, then setting
\begin{equation}
\label{com:com1}
\left(G_R^{(b)}\right)^2=\frac{1}{1+2\SNR}\dot{=}\SNR^{-1},
\end{equation}
for $b=1,2,...,B$, guarantees to satisfy the power constraint at the relay. We conclude that the overall achievable rate region is given by \eqref{eq:AFVSIICancel} and \eqref{eq:AFVSIRegion} subject to the assignment \eqref{com:com1}.

\subsection{Evaluating the DMT of Full-Duplex AF Relaying}
We begin by evaluating the rates and the DMT associated with the transmission of the pair Tx$_1$-Rx$_1$. Note that with AF at the relay, the outage events at consecutive transmission block are correlated. To understand the reason for this, consider Rx$_1$ and note that at for decoding $m_{1,b}$ at block $b+1$, Rx$_1$ uses the Rx-CSI at relay from blocks $b-1$ and $b$, i.e., $\th_3^{b-1}$ and $\th_3^{b}$ (see \eqref{eq:eq32}). For decoding $m_{1,b+1}$ at block $b+2$, Rx$_1$ uses $\th_3^{b}$ and $\th_3^{b+1}$. As the realization $\th_3^{b}$ is used in decoding of both $m_{1,b}$ and $m_{1,b+1}$, then the outage events corresponding to decoding of these two messages are correlated. Let $\outage_b$ denote the outage event at block $b$ at Rx$_1$. Then, the probability of outage in transmission of $B$ blocks is given by
\begin{equation*}
    \Pr(\outage)=\Pr\left(\cup_{b=1}^{b=B-1} \outage_b\right) \le \sum _{b=1}^{b=B-1}\Pr(\outage_b).
\end{equation*}
when the inequality is due to union bound.
Hence, since the probability of outage is upper bounded by the sum of per block outage probabilities, we consider the asymptotical behaviour of the outage probability for a single transmission. First, note that for mutually independent complex Normal channel inputs, i.i.d. in time, we have
\begin{subequations}
\begin{eqnarray}
    \label{eq:eq17}I(X_{1};Y_{2}|\th_2,\th_3)&\dot{\ge}&\log\left(\frac{\SNR^{\alpha-\theta_{12}^{(b)}}}{\SNR+\SNR^{\beta}}\right)\\
     \label{eq:eq20}I\big(X_{1}^{(b)};\tilde{Y}_{1}^{(b)},\tilde{Y}_{1}^{(b+1)}\big|\th_1,\th_3\big)&\dot{\ge}& \log\bigg(\frac{\SNR^{2-\theta_{11}^{(b)}-\theta_{11}^{(b+1)}}}{\SNR+\SNR^{\beta}+\SNR^{2\beta-2}}+\frac{\SNR^{2\beta-1-\theta_{13}^{(b)}-\theta_{31}^{(b)}-\theta_{31}^{(b+1)}}}{\SNR+\SNR^{\beta}+\SNR^{2\beta-2}}\bigg).
\end{eqnarray}
\end{subequations}
Hence, every nonnegative $R_1$ satisfying
\begin{eqnarray*}
    &&\label{eq:R1const}\!\!\!\!\!\!\!\!\!\!R_1\dot{\le} \min\Bigg\{\log\Big(\SNR^{\alpha-\theta_{12}^{(b)}-\max\{1,\beta\}}\Big), \log\Big(\SNR^{2-\max\{1,\beta,2\beta-2\}-\theta_{11}^{(b)}-\theta_{11}^{(b+1)}}+\\
    &&\qquad\qquad\qquad\qquad\qquad\qquad\qquad\qquad\qquad\SNR^{2\beta-1-\max\{1,\beta,2\beta-2\}-\theta_{13}^{(b)}-\theta_{31}^{(b)}-\theta_{31}^{(b+1)}}\Big)\Bigg\},
\end{eqnarray*}
is achievable. We therefore obtain the following DMT region for transmission Tx$_1$-Rx$_1$ at block $b$:
\begin{eqnarray*}
    d_{\textrm{AF}_{FD}}^{-}(r_1)&=&
    \begin{cases}
        \min\Big\{(1-r_1)^+,(\alpha-1-r_1)^+\Big\} & \qquad\qquad ,\beta<1\\
        \min\Big\{(2-\beta-r_1)^++(\beta-1-r_1)^+,(\alpha-\beta-r_1)^+\Big\}  & \qquad\qquad ,1\le\beta<2\\
        \min\Big\{(1-r_1)^+,(\alpha-\beta-r_1)^+\Big\} & \qquad\qquad ,\beta>2
    \end{cases}
\end{eqnarray*}
Since $\Pr(\outage_{b})\doteq\SNR^{-d_{\textrm{AF}_{FD}}^{-}(r_1)}$ is independent of $b$, then $\Pr(\outage)\doteq B\SNR^{-d_{\textrm{AF}_{FD}}^{-}(r_1)}=\SNR^{\log_{\SNR}(B)-d_{\textrm{AF}_{FD}}^{-}(r_1)}\doteq \SNR^{-d_{\textrm{AF}_{FD}}^{-}(r_1)}$.
Following similar steps for Tx$_2$-Rx$_2$, we obtain the DMT region $d_{\textrm{AF}_{FD}}^{-}(r_1,r_2)$ stated in \eqref{eq:af_achievable_dmt}. This completes the proof.
\tend

\noindent
{\bf Comment D.1}. In the high SNR regime, the logarithm in \eqref{eq:eq20} is dominated by the summation of two terms. 
Note that the first term contributes $\big(2-\max\{1,\beta,2\beta-2\}-r_1\big)^+$ to the diversity gain corresponding to transmission Tx$_1$-Rx$_1$, while the second term contributes $\big(2\beta-1-\max\{1,\beta,2\beta-2\}-r_1\big)^+$.
Therefore, $\beta$ does not affect the diversity gain for values of $\beta \le 1$; For $1<\beta<2$, increasing $\beta$ decreases the contribution of the first term due to noise
amplification, and increases the contribution of the second term, and for $\beta>2$ the diversity gain is again independent of $\beta$.
Thus, when the interference is strong enough to allow decoding the interference without decreasing the achievable diversity gain, the maximal achievable diversity gain for the AF-FD scheme is $\min\big\{(1-r_1)^+,(1-r_2)^+\big\}$.

\section{Proof of Theorem \ref{thm:theoremafjd}}
\label{app:AFHDProof}
We first construct a transmission scheme, and then analyze its DMT region.
\vspace{-0.5cm}
\subsection{Overview of the Transmission Scheme}
\label{sec:HDAF1Ach}
We follow the principles of the scheme proposed in \cite{Azarian:05}, which studied HD-AF for the single relay channel: the relay operation is done in consecutive pairs of channel symbols (no overlap). At the first symbol time of each pair, the relay receives the channel output while remaining silent. At the second symbol time, the relay transmits a scaled version of the symbol it received at the first symbol time. Without loss of generality, assume that at time $i$, $i\in\{1,3,5,...,n-1\}$, the relay receives, and at time $i+1$ it transmits.
We also assume that $n$ is even. Tx$_1$ and Tx$_2$ transmit only at the first $n-1$ symbols. At the $n$'th symbol, Tx$_1$ and Tx$_2$ remain silent. The corresponding rate loss is asymptotically negligible. Thus, we have
\begin{eqnarray*}
    Y_{1,i} & = & \sqrt{\SNR}H_{11} X_{1,i} + \sqrt{\SNR^{\alpha}}H_{21} X_{2,i} + Z_{1,i}\\
    Y_{2,i} & = & \sqrt{\SNR^{\alpha}}H_{12} X_{1,i} + \sqrt{\SNR}H_{22} X_{2,i} + Z_{2,i}\\
    Y_{3,i} & = & \sqrt{\SNR}H_{13} X_{1,i} + \sqrt{\SNR}H_{23} X_{2,i} + Z_{3,i}\\
    Y_{1,i+1} & = & \sqrt{\SNR} H_{11} X_{1,i+1} + \sqrt{\SNR^{\alpha}}H_{21} X_{2,i+1}+\sqrt{\SNR^{\beta}}H_{31} G_{R,i} \left(\sqrt{\SNR}H_{13} X_{1,i} + \sqrt{\SNR}H_{23} X_{2,i} + Z_{3,i}\right)  + Z_{1,i+1}\\
    Y_{2,i+1} & = & \sqrt{\SNR^{\alpha}} H_{12} X_{1,i+1} + \sqrt{\SNR}H_{22} X_{2,i+1}+\sqrt{\SNR^{\beta}}H_{32} G_{R,i} \left(\sqrt{\SNR}H_{13} X_{1,i} + \sqrt{\SNR}H_{23} X_{2,i} + Z_{3,i}\right)  + Z_{2,i+1}.
\end{eqnarray*}
The CSI assumptions are the same as those considered in Section \ref{sec:AFregionHD}.  The code construction, encoding and decoding are as follows:

\subsubsection{Code Construction}
    \label{sec:HDAF1code}
        Set $X_k\sim\CN(0,1), k\in\{1,2\}$. For each $m_{k} \in \mathcal{M}_{k}, k\in\{1,2\}$ select a codeword $\xvec_{k}(m_{k})$ according to the p.d.f. $f_{\Xvec_{k}}\big(\xvec_{k}(m_{k})\big)=\prod_{i=1}^n f_{X_{k}}\big(x_{{k},i}(m_{k})\big)$.

\subsubsection{Encoding at the Sources and at the Relay}
    \label{sec:HDAF1encode}
        Tx$_k$ transmits $m_{k}$ using $\xvec_{k}(m_{k}), k\in\{1,2\}$. At each even time index the relay transmits a scaled version of the symbol it received at the previous time index: $X_{3,i+1}=G_{R,i} \left(\sqrt{\SNR}H_{13} X_{1,i} + \sqrt{\SNR}H_{23} X_{2,i} + Z_{3,i}\right)$, where $G_{R,i}$ is set as in \eqref{com:com1} to satisfy the power constraint at the relay.
        At odd time indices the relay does not transmit.

\subsubsection{Decoding at the Destinations}
    \label{sec:HDAF1decode}
    Each receiver jointly decodes $m_{1}$ and $m_{2}$ using a maximum likelihood (ML) decoder. However, note that as the relay transmits a scaled version of its received signal,
    then each odd-indexed channel output is correlated with its subsequent even-indexed channel output. Thus, we consider consecutive pairs of symbols to which we refer as {\em double-symbols}.
    As the codewords are generated i.i.d., it follows that each codeword of length $n$ can be treated a vector of $\frac{n}{2}$ i.i.d.-generated double-symbols and the probability of error can now be calculated using standard ML arguments as in \cite{Azarian:05}.
    For $k\in\{1,2\}$, define $\Xvec^{(D)}_{k,i}\triangleq(X_{k,2i-1}, X_{k,2i})^T, i\in\{1,2,3,...,\frac{n}{2}\}$. It follows that when the $i$'th double-symbol is transmitted by Tx$_k$, the odd-indexed symbol $X_{k,2i-1}$ is transmitted while the relay listens to the sources, and the even-indexed symbol $X_{k,2i}$ is transmitted while the relay transmits. Let $\Yvec^{(D)}_{k,i}$ denote the received signal at Rx$_k$ corresponding to the double-symbol $\Xvec^{(D)}_{k,i}$.
    As discussed above $\left\{\Yvec^{(D)}_{k,i}\right\}_{i=1}^{\frac{n}{2}} \triangleq \Yvec^{(D)}_{k}$ is a vector of the $\frac{n}{2}$ i.i.d. double-symbols received at Rx$_k$.
    Similarly define $\left\{\Xvec^{(D)}_{k,i}\right\}_{i=1}^{\frac{n}{2}} \triangleq \Xvec^{(D)}_{k}$.
    Applying the ML decoding rule, Rx$_k$, $k=1,2$, declares that $(\hat{m}_1,\hat{m}_2)$ was transmitted if
    \begin{equation*}
        (\hat{m}_{1},\hat{m}_{2})=\argmax_{(m_{1}, m_2) \in \mathcal{M}_{1}\times\mathcal{M}_{2}}\!\!\Pr\big(\yvec^{(D)}_{k}\big|\xvec^{(D)}_{1}(m_{1}),\xvec^{(D)}_{2}(m_{2}),\th_k,\th_3\big).
    \end{equation*}
    Using the same notation as in \cite[Section IIV-C]{Tse:04}, we define for each nonempty set $\mathcal{S}\subseteq\{1,2\}$, an error event
        $\mathcal{E}_{\mathcal{S}}\triangleq\{\hat{m}_k\neq m_k,\forall k\in \mathcal{S} \hspace{0.2cm}\mbox{ and }\hspace{0.2cm} \hat{m}_k=m_k,\forall k\in \mathcal{S}^c\}$.
    It follows that the event of decoding error at Rx$_k$ consists of the union of the events $\mE_1$, $\mE_2$ and $\mE_{\{1,2\}}$.
    From \cite[Eqn. (28)]{Tse:04} it follows that the asymptotic probability of error can be evaluated by subtracting from the received signal all signals corresponding
    to the messages in $\mS^c$.


    In the following we evaluate the probability of error at Rx$_1$ by considering each error event and deriving the corresponding probability of error
     for an ML decoder which processes double-symbols:
    \begin{itemize}
        \item $\mathcal{E}_{\{1\}}\triangleq\{\hat{m}_1\neq m_1, \hat{m}_2=m_2\}$:
        Let $(\hat{Y}_{1,2i-1}, \hat{Y}_{1,2i}), i\in\{1,2,3,...,\frac{n}{2}\}$ denote the interference-free signal received at Rx$_1$:
        \begin{eqnarray*}
            \hat{Y}_{1,2i-1} & = & \sqrt{\SNR}H_{11} X_{1,2i-1} +  Z_{1,2i-1}\\
            \hat{Y}_{1,2i} & = & \sqrt{\SNR}H_{11} X_{1,2i} + \sqrt{\SNR^{\beta}}H_{31} G_{R,i} \left(\sqrt{\SNR}H_{13} X_{1,2i-1} + Z_{3,2i-1}\right)  + Z_{1,2i},
        \end{eqnarray*}
        $i\in\{1,2,3,...,\frac{n}{2}\}$. This corresponds to a point-to-point channel whose input is the double-symbol $\Xvec^{(D)}_{1,i}\triangleq(X_{1,2i-1}, X_{1,2i})^T$ and its output is the double-symbol $\hat{\Yvec}^{(D)}_{1,i}\triangleq(\hat{Y}_{1,2i-1},\hat{Y}_{1,2i})^T$.
        Let $\Hmat_{1,i}$, $\Zvec^{(D)}_{1,i}$, $\C_{X_{1,i}^{(D)}}$, and $\C_{Z_{1,i}^{(D)}}$ denote the channel matrix, the noise double-symbol, the covariance matrix of $\Xvec^{(D)}_{1,i}$  and the covariance matrix of the noise at Rx$_1$, respectively:
        \begin{eqnarray*}
        \Hmat_{1,i}&\triangleq&\left[\begin{array}{cc}
        \sqrt{\SNR}h_{11} & \hspace{0.5 cm}0\\
        \hspace{0.5 cm}\sqrt{\SNR^{1+\beta}}h_{13}h_{31}G_{R,i} & \hspace{0.5cm}\sqrt{\SNR}h_{11}\end{array} \right],\;\;\;
        \Zvec^{(D)}_{1,i} \triangleq \left[\begin{array}{c}
        Z_{1,2i-1}\\
        \sqrt{\SNR^{\beta}}h_{31}G_{R,i}  Z_{3,2i-1}+Z_{1,2i}\end{array}\right],
        \end{eqnarray*}
        and
        \begin{eqnarray*}
        \C_{X^{(D)}_{1,i}}&\triangleq&\cov\Big(\Xvec^{(D)}_{1,i}\Big)\triangleq\left[\begin{array}{cc}
        1 &\;\; 0\\
        0 &\;\; 1 \end{array}\right]\equiv \C_{X^{(D)}_{1}},\;
        \C_{Z^{(D)}_{1,i}} \triangleq \cov\Big(\Zvec^{(D)}_{1,i}\Big)\triangleq\left[\begin{array}{cc}
        1 & \hspace{0.3 cm}0\\
        0 & \hspace{0.3 cm}1+\SNR^{\beta}|h_{31}|^2G_{R,i}^2\end{array}\right]\!.
        \end{eqnarray*}
        Hence, $\tilde{\Yvec}^{(D)}_{1,i}=\Hmat_{1,i}\Xvec^{(D)}_{1,i}+\Zvec^{(D)}_{1,i}$.
         Note that the assignment $G_{R,i}^2=\SNR^{-1}$ (\cite[Appendix B]{Azarian:05}) satisfies the power constraint at the relay (see, Eqn. \eqref{com:com1}).
        Setting $G_{R,i}^2=\SNR^{-1}$ for all $i=1,3,5,...,n-1$, we obtain that $\Hmat_{1,i}$ and $\C_{Z^{(D)}_{1,i}}$ do not depend on $i$, hence we denote $\Hmat_{1,i}\equiv\Hmat_{1}$ and $\C_{Z^{(D)}_{1,i}}\equiv\C_{Z^{(D)}_{1}}$.
        Next, following similar steps as those used in \cite{Azarian:05}, we conclude that an upper bound on pairwise error probability (PEP) for the ML decoding rule associated with $\mathcal{E}_{\{1\}}$ is
        \begin{equation}
        \label{eq:PEP}
            P_{{PE}_1}\leq \left(\det\Big(\Imat_2+\frac{1}{2}\Hmat_1\C_{X_1}^{(D)}\Hmat_1^H\left(\C^{(D)}_{Z_1}\right)^{-1}\Big)\right)^{-\frac{n}{2}}.
        \end{equation}
        Plugging $\Hmat_1$, $\C_{X_1}^{(D)}$ and $\C_{Z_1}^{(D)}$ into \eqref{eq:PEP} 
        we obtain
        \begin{eqnarray}
            \label{eq:eq23}P_{{PE}_1}&\leq& \left(\!\!1\!+\!\frac{1}{2}\SNR^{1-\theta_{11}}\! +\! \frac{\frac{1}{2}\SNR^{\beta-\theta_{31}-\theta_{13}}}{1+\SNR^{\beta-\theta_{31}-1}} \!+\! \frac{\frac{1}{4}\SNR^{2-2\theta_{11}}}{1+\SNR^{\beta-\theta_{31}-1}} + \frac{\frac{1}{2}\SNR^{1-\theta_{11}}}{1+\SNR^{\beta-\theta_{31}-1}}\right)^{-\frac{n}{2}}\\
            &\stackrel{(a)}{\leq}& \left(1+\frac{1}{2}\SNR^{1-\theta_{11}} + \frac{\frac{1}{2}\SNR^{\beta-\theta_{31}-\theta_{13}}}{1+\SNR^{\beta-1}} + \frac{\frac{1}{4}\SNR^{2-2\theta_{11}}}{1+\SNR^{\beta-1}} + \frac{\frac{1}{2}\SNR^{1-\theta_{11}}}{1+\SNR^{\beta-1}}\right)^{\!\!\!-\frac{n}{2}}\nonumber\\
            &\stackrel{(b)}{\dot{\le}}&\begin{cases}
            \left(\SNR^{\beta-\theta_{31}-\theta_{13}} + \SNR^{2-2\theta_{11}}\right)^{-\frac{n}{2}} & \!\!\!\!,\beta \!\le\!1\nonumber\\
            \min\!\!\left\{\!\!\left(\SNR^{1-\theta_{11}}\!\!+\SNR^{1-\theta_{31}-\theta_{13}}\right)^{\!-\frac{n}{2}}\!\!,\! \left(\SNR^{1-\theta_{31}-\theta_{13}}\!\!+\SNR^{3-\beta-2\theta_{11}}\right)^{\!-\frac{n}{2}}\right\}& \!\!\!\! ,\beta \!>\!1\;\;\;\;\;\;\;\;
            \end{cases}
        \end{eqnarray}
        where (a) follows as $\theta_{31}\ge0$ and therefore, omitting $\theta_{31}$ from the exponents in the denominators increases the dominators and therefore, decreases the expression in the parentheses in \eqref{eq:eq23} and eventually increases the entire expression;  (b) follows since
        for $\beta \le1$ we obtain $1+\SNR^{\beta-1}\doteq1$  as $\SNR\rightarrow\infty$. Recall that the target rate at Tx$_1$ is $r_1\log\SNR$ bits per channel uses.
        As the target rate of the double-symbols is set to $R^{(D)}_{1,T}=r^{(D)}_1\log(\SNR)$ bits {\em per two channel uses}, then we have a total of $\SNR^{\frac{n}{2}r^{(D)}_1}=\SNR^{\frac{n}{2}2 r_1}$ codewords. Hence, applying the union bound over all the codewords, the probability of error in decoding the message from Tx$_1$ at Rx$_1$ can be upper bounded by
        \begin{equation}
            \label{eq:PE}
            \!\!\Pr(\mE_1)\dot{\leq}\!\!\begin{cases}
            \!\SNR^{-\frac{n}{2}\left[\big(\max\{\beta-\theta_{31}-\theta_{13},2-2\theta_{11}\}\big)^+-2r_1\right]} & \!\!\!\!\!,\beta\!\le\!1\\
            \!\SNR^{-\frac{n}{2}\left[\max\left\{\!\!\big(\max\{1-\theta_{11},1-\theta_{31}-\theta_{13}\}\big)^+\!\!\!
                    -2r_1,\big(\max\{1-\theta_{31}-\theta_{13},3-\beta-2\theta_{11}\}\big)^+
                    \!\!\!-2r_1\!\!\right\}\right]} &\!\!\!\!\!,\beta\!>\!1 \end{cases}
        \end{equation}

        \item $\mathcal{E}_{\{2\}}\triangleq\{\hat{m}_1= m_1, \hat{m}_2\neq m_2\}$: Note that as decoding $m_{2}$ is not required at Rx$_1$, then an outage corresponding to $\mathcal{E}_{\{2\}}$ need not be accounted for at Rx$_1$, and therefore $\mE_{\{2\}}$ does not constrain the achievable DMT region at Rx$_1$.

        \item $\mathcal{E}_{\{1,2\}}\triangleq\{\hat{m}_1\neq m_1, \hat{m}_2\neq m_2\}$: Define the super-symbol $\Xvec^{(S)}_{1,i}\triangleq(X_{1,2i-1}, X_{2,2i-1}, X_{1,2i}$, $X_{2,2i})^T$, $i\in\{1,3,5,...,\frac{n}{2}\}$ as a vector of two consecutive pairs of symbols transmitted by Tx$_1$ and Tx$_2$.
            The corresponding received signal at Rx$_1$ is
        \begin{eqnarray*}
            Y_{1,2i-1} & = & \sqrt{\SNR} H_{11} X_{1,2i-1} + \sqrt{\SNR^\alpha}H_{21} X_{2,2i-1} + Z_{1,2i-1}\\
            Y_{1,2i}   & = & \sqrt{\SNR} H_{11} X_{1,2i} + \sqrt{\SNR^\alpha}H_{21} X_{2,2i} \\
             &&\quad+\sqrt{\SNR^{\beta}}H_{31} G_{R,i} \left(\sqrt{\SNR}H_{13} X_{1,2i-1}+ \sqrt{\SNR}H_{23} X_{2,2i-1} + Z_{3,2i-1}\right)  + Z_{1,2i}.
        \end{eqnarray*}
        Next, define $\Yvec^{(D)}_{1,i}\triangleq(Y_{1,2i-1},Y_{1,2i})^T$ and note that $\left\{\Yvec^{(D)}_{1,i}\right\}_{i=1}^{\frac{n}{2}}$ are i.i.d. Define
        \begin{eqnarray*}
        \Hmat_{1,i}&\triangleq&\left[\begin{array}{cccc}
        \sqrt{\SNR}h_{11} & \hspace{0.5 cm}\sqrt{\SNR^\alpha}h_{21}&\hspace{0.5 cm}0&\hspace{0.5 cm}0\\
        \sqrt{\SNR^{1+\beta}}h_{13} G_{R,i} h_{31}& \hspace{0.5 cm}\sqrt{\SNR^{1+\beta}}h_{23} G_{R,i} h_{31}& \hspace{0.5 cm}\sqrt{\SNR}h_{11}
                & \hspace{0.5 cm}\sqrt{\SNR^\alpha}h_{21} \end{array} \right]\\
        \Zvec_{1,i}^{(D)}&\triangleq&\left[\begin{array}{c}
        Z_{1,2i-1} \\
        \sqrt{\SNR^{\beta}}h_{31}G_{R,i}Z_{3,i-1}+ Z_{1,2i}\end{array}\right]\!, \quad\;\;
        \C_{Z^{(D)}_{1,i}}\triangleq\cov\Big(\Zvec_{1,i}^{(D)}\Big)=\left[\begin{array}{cc}
        1 & \hspace{0.5 cm}0\\
        0 & \hspace{0.5 cm}1+\SNR^{\beta-\theta_{31}}G_{R,i}^2\end{array}\right]\!,
        \end{eqnarray*}
        $\C_{X_{1,i}^{(S)}} \triangleq \cov\Big(\Xvec^{(S)}_{1,i}\Big)=\Imat_4 \equiv \C_{X_{1}^{(S)}}$.
        Hence, $\Yvec^{(D)}_{1,i}=\Hmat_{1,i}\Xvec^{(S)}_{1,i}+\Zvec^{(D)}_{1,i}$. Setting $G_{R,i}^2=\SNR^{-1}$ we satisfy the power constraint at the relay, and obtain that $\Hmat_{1,i}$ and $\C_{Z^{(D)}_{1,i}}$
        are independent of $i$, thus $\Hmat_{1,i}\equiv \Hmat_{1}$ and $\C_{Z^{(D)}_{1,i}} \equiv \C_{Z^{(D)}_{1}}$.
        Following steps similar to those used
        in \cite{Azarian:05}, we obtain that an upper bound on the PEP associated with $\mathcal{E}_{\{1,2\}}$ is given by
        \begin{equation}
            \label{eq:PEP2}
            P_{{PE}_{12}}\leq \det\left(\Imat_2+\frac{1}{2}\Hmat_1\C_{X_1}^{(S)}\Hmat_1^H\left(\C_{Z_1}^{(D)}\right)^{-1}\right)^{-\frac{n}{2}}.
        \end{equation}
        Plugging $\Hmat_1$, $\C_{X_1}^{(S)}$ and $\C_{Z_1}^{(D)}$ into \eqref{eq:PEP2} we obtain
        \begin{eqnarray}
            \!\!\!\!P_{{PE}_{12}}&\leq& \bigg(1+\frac{1}{2}\big(\SNR^{1-\theta_{11}}+\SNR^{\alpha-\theta_{21}}\big) + \frac{\frac{1}{2}\big(\SNR^{1-\theta_{11}}+\SNR^{\alpha-\theta_{21}}\big)}{1+\SNR^{\beta-\theta_{31}-1}} +\nonumber\\
            &&\quad\frac{\frac{1}{4} \big(\SNR^{1-\theta_{11}}+\SNR^{\alpha-\theta_{21}}\big)^2}{1+\SNR^{\beta-\theta_{31}-1}}+\frac{\frac{1}{2}\SNR^{\beta-\theta_{31}-1}\big(\SNR^{1-\theta_{13}}+\SNR^{1-\theta_{23}}\big)} {1+\SNR^{\beta-\theta_{31}-1}}+\nonumber\\
            \label{eq:eq21}&&\quad\qquad\qquad\qquad\frac{\frac{1}{4}\SNR^{\beta-\theta_{31}-1}\big|\sqrt{\SNR^{1+\alpha}}h_{13}h_{21}-\SNR h_{23}h_{11}\big|^2}{1+\SNR^{\beta-\theta_{31}-1}}\bigg)^{-\frac{n}{2}}\\
            &\stackrel{(a)}{\le}&\quad \bigg(\frac{\frac{1}{4} \big(\SNR^{1-\theta_{11}}+\SNR^{\alpha-\theta_{21}}\big)^2}{1+\SNR^{\beta-\theta_{31}-1}} +\frac{\frac{1}{2}\SNR^{\beta-\theta_{31}-1}\big(\SNR^{1-\theta_{13}}+\SNR^{1-\theta_{23}}\big)}{1+\SNR^{\beta-\theta_{31}-1}}\bigg)^{ -\frac{n}{2}}\nonumber\\
            \label{eq:eq26}
            &\stackrel{(b)}{\le}&\quad \bigg(\frac{\frac{1}{4} \big(\SNR^{1-\theta_{11}}+\SNR^{\alpha-\theta_{21}}\big)^2}{1+\SNR^{\beta-1}} +\frac{\frac{1}{2}\SNR^{\beta-\theta_{31}-1}\big(\SNR^{1-\theta_{13}}+\SNR^{1-\theta_{23}}\big)}{1+\SNR^{\beta-1}}\bigg)^{ -\frac{n}{2}}\\
            &\stackrel{(c)}{\dot{\le}}&\quad\nonumber\begin{cases}
            \bigg(\SNR^{2-2\theta_{11}}+\SNR^{2\alpha-2\theta_{21}}+\SNR^{\beta-\theta_{31}-\theta_{13}}\bigg)^{-\frac{n}{2}} & \quad ,\beta\le1\\
            \bigg(\SNR^{3-\beta-2\theta_{11}}+\SNR^{2\alpha+1-\beta-2\theta_{21}}+\SNR^{1-\theta_{31}-\theta_{13}}\bigg)^{-\frac{n}{2}} &\quad ,\beta>1,
            \end{cases}
        \end{eqnarray}
        where\! (a)\! follows since the expression in the parentheses\! in \eqref{eq:eq21} is a summation of nonnegative terms and thus,\! removing nonnegative terms from this summation increases the expression on the right hand side of the inequality;\! (b)\! follows since\! $\theta_{31}\ge 0$; and (c) is obtained by omitting\! nonnegative\! terms\! from\! \eqref{eq:eq26}.\!\! Using similar steps as those used to evaluate the probability of $\mathcal{E}_{\{1\}}$, set the target rate to $R_{k,T}\!=\!r_k\log(\SNR), k\in\{1,2\}$, bits per channel use.\! Thus, when decoding both messages at Rx$_1$, there is a total of $\SNR^{\frac{n}{2}2(r_1+r_2)}$ possible codewords. By applying the union bound over all the codewords we conclude that
        \begin{equation}
            \label{eq:PE2}
            \Pr(\mE_{\{1,2\}})\dot{\leq}\begin{cases} \SNR^{-\frac{n}{2}\Big[\max\big\{(2-2\theta_{11}),(2\alpha-2\theta_{21}),(\beta-\theta_{31}-\theta_{13})\big\}-2r_1-2r_2\Big]^+} & \;\;\;,\beta\le1\\
            \SNR^{-\frac{n}{2}\Big[\max\big\{(3-\beta-2\theta_{11}),(2\alpha+1-\beta-2\theta_{21}),(1-\theta_{31}-\theta_{13})\big\}-2r_1-2r_2\Big]^+} & \;\;\;,\beta\ge1
            \end{cases}
        \end{equation}
    \end{itemize}

 \vspace{-0.5 cm}

\subsection{Evaluating the DMT region of Half-Duplex AF Relaying}
\!\!\!First,\! we evaluate the DMT region corresponding to $\mathcal{E}_{\{1\}}$: Define $\!P_{\outage,1}$ as the probability of the event in which the channel realizations are s.t. the probability of error
corresponding to $\mathcal{E}_{\{1\}}$
cannot be made arbitrarily small. Following similar arguments as those in \cite[\!Proof of Thm. 3]{Azarian:05}, we conclude that $P_{\outage,1}$ can be upper bounded by $P_{\outage,1}\dot{\leq} \SNR^{-d_{HD_1}(r_1)}$ where $d_{HD_1}(r_1)$ is obtained from 
\eqref{eq:PE} as follows:
For $\beta \le 1$ 
the maximal $\Pr(\mE_{\{1\}})$ can be obtained from
the following minimization problem:

 \vspace{-0.9 cm}

\begin{eqnarray*}
     &&\min \theta_{11}+\theta_{31}+\theta_{13}\\
     &&\mbox{s.t. }\max\big\{(\beta-\theta_{31}-\theta_{13}),(2-2\theta_{11})\big\}<2r_1, \qquad \theta_{11}\ge0,\theta_{31}\ge0,\theta_{13}\ge0.
\end{eqnarray*}
This results in the following DMT for $\beta \le 1$:
\begin{equation}
\label{eq:eq271}
d_1(r_1)=(1-r_1)^++(\beta-2r_1)^+.
\end{equation}
Next, 
for $\beta>1$ we observe that the DMT region is given as the maximum of two expressions. The first expression is $d_2(r_1)=\theta_{11}^*+\theta_{31}^*+\theta_{13}^*$, where $(\theta_{11}^*,\theta_{31}^*,\theta_{13}^*)$ are the optimal arguments of the minimization problem:
\begin{eqnarray*}
     &&\min \theta_{11}+\theta_{31}+\theta_{13}\\
     &&\mbox{s.t. }\max\big\{(1-\theta_{11}),(1-\theta_{31}-\theta_{13})\big\}<2r_1, \qquad \theta_{11}\ge0,\theta_{31}\ge0,\theta_{13}\ge0,
\end{eqnarray*}
for which we obtain the DMT region: 
\begin{equation}
\label{eq:eq272}
d_{2}(r_1)=2(1-2r_1)^+.
\end{equation}
 The second expression is given by $d_3(r_1)=\theta_{11}^*+\theta_{31}^*+\theta_{13}^*$, where $(\theta_{11}^*,\theta_{31}^*,\theta_{13}^*)$ are the optimal arguments
 of the optimization problem:
\begin{eqnarray*}
     &&\min \theta_{11}+\theta_{31}+\theta_{13}\\
     &&\mbox{s.t. }\max\big\{(1-\theta_{31}-\theta_{13}),(3-\beta-2\theta_{11})\big\}<2r_1,\qquad \theta_{11}\ge0,\theta_{31}\ge0,\theta_{13}\ge0,
\end{eqnarray*}
which results in the DMT region:
\begin{equation}
\label{eq:eq273}
d_3(r_1)=(1-2r_1)^+ + \left(\frac{3-\beta}{2}-r_1\right)^+.
\end{equation}
Combining \eqref{eq:eq271}, \eqref{eq:eq272}, and \eqref{eq:eq273} we conclude that $d_{HD_1}(r_1)$ is given by:
\begin{equation}
\label{eq:eq19}
d_{HD_1}(r_1)=\begin{cases}
(1-r_1)^++(\beta-2r_1)^+ & \qquad \qquad ,\beta\le1\\
\max\left\{2(1-2r_1)^+,(1-2r_1)^++(\frac{3-\beta}{2}-r_1)^+\right\} & \qquad \qquad, \beta>1
\end{cases}
\end{equation}
Using similar arguments we obtain the DMT corresponding to the error event $\mathcal{E}_{\{1,2\}}$ from \eqref{eq:PE2}:
\begin{equation}
\label{eq:JointDMT}
d_{HD_{12}}(r_1,r_2)=\begin{cases}
(1-r_1-r_2)^++(\alpha-r_1-r_2)^++(\beta-2r_1-2r_2)^+ & \;,\beta\le1\\
(\frac{3-\beta}{2}-r_1-r_2)^++(\frac{2\alpha+1-\beta}{2}-r_1-r_2)^++(1-2r_1-2r_2)^+ & \;,\beta>1
\end{cases}
\end{equation}
Repeating the same derivations for decoding at Rx$_2$ and combining with \eqref{eq:eq19} and \eqref{eq:JointDMT} we obtain \eqref{eq:HDAFDMT1}. This completes the proof.
\tend


\end{document}